\documentclass[journal]{IEEEtran}
\usepackage{amsmath,amsfonts,amssymb}
\usepackage{amsthm}
\usepackage{romannum}
\usepackage{algorithmic}
\usepackage{array}
\usepackage[caption=false,font=normalsize,labelfont=sf,textfont=sf]{subfig}
\usepackage{textcomp}
\usepackage{xcolor}
\usepackage{stfloats}
\usepackage{url}
\usepackage{verbatim}
\usepackage{bm}
\usepackage{graphicx}
\usepackage{algorithm}
\usepackage{cite}
\usepackage{orcidlink}
\usepackage{balance}

\usepackage{hyperref}
\hypersetup{hidelinks}

\newtheorem{lemma}{Lemma}[section]
\newtheorem{theorem}{Theorem}[section]

\def\BibTeX{{\rm B\kern-.05em{\sc i\kern-.025em b}\kern-.08em
T\kern-.1667em\lower.7ex\hbox{E}\kern-.125emX}}

\begin{document}
\title{Semantic Communication for 6G Networks: A Trade-off between Distortion Criticality and Information Representability}
\author{Faizan Shafi \orcidlink{0009-0004-6632-9746} 
\and, 
Rahul Jashvantbhai Pandya \orcidlink{0000-0002-3259-817X} 
\and, 
Christo Kurisummoottil Thomas \orcidlink{0000-0002-0926-5211} 
\and,
Sridhar Iyer \orcidlink{0000-0002-8466-3316}
\vspace{-8mm}

\thanks{

(\textit{Corresponding author: Faizan Shafi})

Faizan Shafi and Rahul Jashvantbhai Pandya are with the Department of Electrical, Electronics and Communications Engineering, Indian Institute of Technology Dharwad, Dharwad, Karnataka 580011, India (e-mail:
ee23dp004@iitdh.ac.in; rpandya@iitdh.ac.in). Christo Kurisummoottil Thomas is with the Department of Electrical and Computer Engineering, Worcester Polytechnic Institute, MA, USA (e-mail: cthomas2@wpi.edu).
Sridhar Iyer is with the Department of ECE, SGBIT, Belagavi, Karnataka 590010, India (e-mail: sridhariyer@sgbit.edu.in).

This work was supported in part by the Department of Telecommunication (DoT), Ministry of Communications, Government of India under the Telecom Technology Development Fund (TTDF) the scheme implemented through TCOE India under the grant TTDF/6G/48 and IGSTC-04918.}}

\maketitle
\thispagestyle{empty}
\pagenumbering{arabic}
\begin{abstract}
In this work, a self-attention based conditional generative adversarial network (SA–cGAN) framework for the sixth generation (6G) semantic communication system is proposed, explicitly designed to balance the trade-off between distortion criticality and information representability under varying channel conditions. The proposed SA–cGAN model continuously learns compact semantic representations by jointly considering semantic importance, reconstruction distortion, and channel quality, enabling adaptive selection of semantic tokens for transmission. A knowledge graph is integrated to preserve contextual relationships and enhance semantic robustness, particularly in low signal-to-noise ratio (SNR) regimes. The resulting optimization framework incorporates continuous relaxation, submodular semantic selection, and principled constraint handling, allowing efficient semantic resource allocation under bandwidth and multi-constraint conditions. The simulation results show that, although SA–cGAN achieves modest syntactic bilingual evaluation understudy scores at low SNR to approximately 0.72 at 20 dB, it significantly outperforms conventional and JSCC based schemes in semantic metrics, with semantic similarity, semantic accuracy, and semantic completeness consistently improving above 0.90 with SNR. Additionally, the model exhibits adaptive compression behavior, aggressively reducing redundant content while preserving critical semantic information to maintain fidelity. The convergence of training loss further validates the stable and efficient learning of semantic representations. Overall, the results confirm that the proposed SA–cGAN model effectively captures distortion invariant semantic representations and dynamically adapts transmitted content based on distortion criticality and information representability for meaning centric communication in future 6G networks.
\end{abstract}


\section{Introduction}
\IEEEPARstart{T}{he} sixth generation (6G) communication networks aim to achieve unprecedented objectives, including ultra-high data rates, improved energy efficiency, ultra-high user connectivity density, reliability, and ultra-low latency\cite{10540315},\cite{Wheeler_2023},\cite{9663101}. However, Shannon's information theory establishes fundamental limits on reliable bit-level transmission over a communication channel, primarily focusing on maximizing data rates within these theoretical bounds\cite{Shannon1948} and \cite{Tse2005}. In contrast, with the rapid increase in internet users, internet of things (IoT) devices, and data-intensive intelligent applications, future 6G networks must address not only higher data demands but also reliability and resource efficiency requirements. Rather than trying to exceed Shannon’s limits, emerging research explores alternative communication paradigms that redefine the optimization objective itself \cite{Strinati2021} and \cite{9937052}. In this context, semantic communication (SemCom) has gained significant attention as a promising approach that shifts the focus from accurate transmission of raw bits to preservation and reconstruction of meaning at the receiver \cite{6004632},\cite{8461983},\cite{Qin2022Semantic}. By optimizing communication resources with respect to semantic relevance rather than symbol-level fidelity, such systems address a fundamentally different objective and evaluation criterion\cite{Luo2022}. This shift in perspective motivates the exploration of semantic-aware encoding, transmission, and decoding strategies for next-generation wireless systems. In SemCom systems, a key challenge lies in bridging the gap between high-level semantic meaning and the low-level physical representations required for transmission \cite{10538233},\cite{9763856},\cite{10584461}. Existing approaches often struggle to ensure that abstract semantic concepts are represented in a form that is both interpretable and efficiently processable by learning based communication modules, particularly under bandwidth and channel constraints\cite{11276859} and \cite{10603423}. Moreover, preserving semantic intent while adapting to dynamically varying channel conditions remains nontrivial, as excessive compression may lead to semantic loss, whereas conservative transmission incurs significant resource overhead \cite{Thomas2022NeuroSymbolicAI} and \cite{9832831}. These challenges further complicate 6G networks, including ultra-dense IoT deployments, autonomous systems, and intelligent services, where devices exhibit diverse capabilities, channel conditions, and semantic requirements \cite{Xie2021} and \cite{8869705}. Addressing these issues requires semantic-aware communication paradigms that can intelligently align semantic importance with resource availability while remaining robust to semantic distortion and channel uncertainty. Such semantic-aware and learning-driven communication paradigms are increasingly recognized as fundamental enablers for emerging 6G use cases\cite{8766143} and \cite{10494374}. 

\subsection{Related Works}
In recent years, extensive research has explored SemCom and artificial intelligence (AI) enabled transmission frameworks as promising alternatives to traditional bit-centric designs for supporting emerging intelligent and data-intensive applications \cite{Iyer2023SemanticSurvey} and \cite{Atzori2010}. The authors in \cite{9885016}, have proposed a SemCom system called MR\_DeepSC, leveraging deep neural networks (DNN) to enhance performance in multi-user scenarios. They have adopted a semantic recognizer, which is based on a pretrained model, to differentiate users and have made use of transfer learning (TL) as a way of training receiver networks faster. The simulation findings have highlighted the increased performance of the system in different channel environments, which is most prominent in the situation where low signal-to-noise ratios (SNR) prevail. However, there is limited discussion on the potential limitations of the TL approach used to train receiver networks, which could affect the generalizability of the model. To overcome spectral scarcity in 6G networks, \cite{10005215} explores the potential of SemCom to improve transmission reliability and efficiency in 6G networks, particularly in interference prone environments. The study proposes a matching game framework for resource allocation (RA) to optimize bandwidth efficiency. However, the work does not fully explore the impact of varying channel quality conditions on the performance of the SemCom system, which is crucial for optimizing RA schemes. In addition, the authors have not addressed the potential limitations of the semantic extraction and retrieval process, particularly in terms of the accuracy of the similarity coefficient used to evaluate performance. 

The study in \cite{9955525} presents an overview of semantic and task-oriented communications, highlighting the integration of message semantics and communication goals in system design. It discusses the role of information theory and deep learning (DL) in enhancing communication efficiency across various sources of information such as images, text, and video. However, there is a lack of semantics-aware decision-making in real-time reconstruction policies, which could enhance the efficiency of communication systems by considering the context of information at both the source and the receiver. It does not address the implications of distributed parameter estimation under rate constraints, leaving a gap in understanding how to optimize performance in such scenarios. In turn, \cite{9991159} discusses the evolution of communication technologies from a technical to a semantic level, proposing a novel model for SemCom using conceptual spaces and functional compression. It does not explore the implications of semantic errors in detail, particularly how they affect overall communication efficiency in various environments, which is crucial for advancing SemCom systems. The exploration of equivalence relations and efficient semantic encodings is noted as an area that requires further investigation, indicating that the current study does not fully address this aspect.

Furthermore, \cite{10100737} proposed an end-to-end SemCom based image transmission system that considers the physical characteristics of the channel. The proposed method uses semantic segmentation at the transmitter and a pre-trained generative adversarial network (GAN) at the receiver to reconstruct realistic images, trained on the COCO-Stuff dataset. By transmitting only the segmentation map instead of the full image, the system achieves notable bandwidth savings compared to conventional methods. The study also examines the effects of channel distortions and quantization noise, offering valuable insights into bridging SemCom with real world transmission challenges. The authors of \cite{8985539}, have proposed the use of conditional generative adversarial networks (cGANs) to model unknown or complex wireless channels in an end-to-end learning framework, allowing transmitter gradients to flow through a learned channel surrogate. This demonstrates the power of cGANs for channel modeling and DNN based transceiver design, but it focuses on channel emulation rather than semantic adaptation or knowledge graph (KG) aware token selection. The Deep joint source channel coding (JSCC) family \cite{8683463}, develops learned joint source–channel encoders/decoders that significantly improve symbol-level reconstruction (images/text) under noisy channels. These methods excel at end-to-end reconstruction and rate adaptation, but they remain largely symbol-centric, they do not explicitly model semantic importance, KGs, or distortion vs. representability trade-offs that determine which semantic units to send. Several studies \cite{XIE2023102214} use GAN variants (WGAN, WGAN-GP) to synthesize realistic channel realizations or to generate channel statistics for training. These show GANs effectiveness in capturing channel distributions, but they focus on channel modeling or dataset synthesis rather than on semantic encoding/decoding and KG-aligned, attention-guided semantic selection for end-to-end meaning preservation. The study in \cite{Xu2023KnowledgeEnhancedSC} formulates KG based semantic transmission schemes that extract and transmit graph structured semantics to improve interpretation at the receiver. These schemes show how KG alignment improves semantic robustness, yet they typically do not combine KG guidance with adaptive generative models and self-attention token selection under a joint optimization framework. To summarize, the authors have used KG for information representation, but not co-optimized with cGAN driven RA and submodular selection. Recent works  \cite{10843783} investigate attention mechanisms and importance-aware selection to prioritize tokens or features for semantic transmission. These works demonstrate that attention can drive semantic compression, but they generally stop short of integrating a generative channel-aware model for dynamic symbol generation and do not provide a principled optimization treatment balancing distortion criticality and information representability for RA. 

As is evident, no prior work has explored self-attention based conditional generative adversarial network (SA-cGAN) enabled RA strategies for 6G SemCom. Existing SemCom studies focus primarily on semantic encoding and reconstruction, but they often overlook the problem of adaptive semantic RA under physical channel constraints. In particular, most existing approaches do not explicitly model the trade-off between semantic fidelity, reconstruction distortion, and channel reliability when selecting which semantic elements should be transmitted. Moreover, while attention-based models can identify important semantic tokens, they typically lack a generative mechanism to adapt the transmitted semantic representations to varying channel conditions. Motivated by these limitations, this work proposes a SA-cGAN driven RA methodology for 6G wireless networks, specifically addressing the trade-offs between information representability and distortion criticality (detailed in section II-A).
\vspace{-4mm}
\subsection{Key Contributions}
The main contribution of this work is the development of a novel SemCom framework that explicitly balances distortion criticality and information representability under practical 6G channel constraints using a SA–cGAN model. By integrating knowledge driven semantic extraction with adaptive generative modeling, the proposed framework enables efficient and meaning centric communication as an alternative to conventional bit level transmission approaches. Moreover, cGAN's unique capability for feedback based dynamic adaptation of semantic representation will provide efficient RA. The key contributions of the current article are as follows:
\begin{enumerate}
\item We propose a SemCom system that utilizes a KG oriented SA-cGAN to encode and decode information semantics, thereby reducing semantic ambiguity. This operation ensures that semantic selection is not solely driven by statistical attention, but is explicitly constrained by symbolic semantic structure encoded in the KG. 
\item We introduce an intelligent, optimized, and performance enhanced SA-cGAN enabled resilient system model with adaptive RA strategies for 6G SemCom networks. In the current study, with regard to RA, we mainly focus on bandwidth as the primary resource.
\item We furthermore propose the design of an SA-cGAN based rate adaptation mechanism that dynamically adjusts transmission rates based on both semantic content and network conditions addressing the trade-offs between information representability and distortion criticality. The framework supports multiple adaptation strategies, including conservative (high distortion, low representability), moderate (medium distortion and representability), and aggressive approach (low distortion, high representability).
\item We solve the dynamic and rigorous optimization problem via continuous relaxation, submodular selection, Lagrangian dual formulation, and Karush–Kuhn–Tucker (KKT) based characterization, thereby explicitly coupling information representability and distortion criticality with semantic RA.
\item Extensive simulation results demonstrate that the proposed SA–cGAN framework achieves substantial gains in semantic oriented metrics. In particular, semantic similarity improves from $\approx$ 0.23 to over 0.93, while semantic accuracy and semantic completeness increase consistently from around 0.18 to above 0.90, validating robust semantic transmission under varying channel conditions.
\end{enumerate} 

The rest of this paper is structured as follows: section II discusses the system model, while section III introduces the performance metrics and formulates the optimization problem. In section IV, solution for optimization problem is detailed followed by section V, which presents and discusses the obtained results. Finally, section VI concludes the paper. The key notations used throughout the article are expanded in Table I. 

\begin{table}
    \caption{Key Notations}
    \resizebox{\columnwidth}{!}{%
    \begin{tabular}{p{1.4cm}|p{5.4cm}}
    \hline
    Notation  &  Interpretation\\
    \hline
    $\mathcal{H}_{s}(\cdot)$ & Semantic encoder\\
    $\mathcal{H}_{c}(\cdot)$ & Channel encoder\\
    $\mathcal{C}_{\text{SA--cGAN}}(\cdot)$ & SA-cGAN model \\
    $X_i$ & Original text required to be sent to user $i$\\
    $\hat{X_i}$ & Decoded/received text\\
    $V$ & Global vocabulary or token set \\ 
    $x_{i,n}$ & Token $n$ in textual data $X_i$ \\
    $N_i$ & Total number of tokens in textual data $X_i$\\
    $e_{i,j}$ & Entity $j$ in text $X_i$\\
    $e_{i,k}$ & Entity $k$ in text $X_i$\\
    $r_{i,jk}$ & Relation between entity $e_{i,j}$ and entity $e_{i,k}$ \\
    $E_i$ & Number of entities in text $X_i$\\
    $R_i$ & Number of relations in text $X_i$\\
    $\mathcal{G}_i$ & Semantic information of textual data $X_i$ \\
    $G_i$ & Number of semantic triples in $\mathcal{G}_i$\\
    $\varepsilon_{i}^g$ & Semantic triples $g$ in $\mathcal{G}_i$\\
    $v_{i,b}^g$ & Token  b in semantic triple $\varepsilon_{i}^g$\\
    $S_{i,j}^g $ & Number of tokens in entity $e_{i,j}^g$\\
    $S_{i,jk}^g$ & Number of tokens in relation $r_{i,jk}^g$\\
    $\mathcal{E}$ & Entity\\
    $\mathcal{R}$ & Relation\\
    $\mathcal{K}$ & Knowledge graph\\
    $r_k$ & Semantic symbol\\
    $S$ & Selected semantic token set \\
    $\alpha_i$ & Self-attention/semantic importance score\\
    $\tau_\alpha$ & Semantic importance threshold\\
    $s_i$ & Binary semantic selection variable\\
    $G(\cdot)$ & cGAN generator\\
    $D(\cdot)$ & cGAN decoder\\
    $P_{\mathrm{src}(\cdot)}$ & Probability distribution over the set of semantic $\mathcal{R}$ in source KG\\
    $P_{\mathrm{dest}(\cdot)}$ & Probability distribution over the set of semantic $\mathcal{R}$ in destination KG\\
    ${B}_{\Phi}$ & BERT model \\
    $\log \mathrm{BLEU}_i$ & BLEU score\\
    $\mathrm{Sim}_i$ & Semantic similarity\\
    $\mathrm{Acc}_i$ & Semantic accuracy\\
    $\mathrm{Comp}_i$ & Semantic completeness \\
    $\mathcal{L}_{\mathrm{IR}}$ & Information representability loss\\
    $\mathcal{L}_{\mathrm{DC}}$ & Distortion criticality loss\\
    \hline
    \end{tabular}
}
 \end{table}
\vspace{-4mm}
\section{System Model}
We consider a point-to-point SA-cGAN enabled 6G SemCom model, where a source node transmits textual information to a destination node over a wireless channel, as illustrated in Fig.~\ref{System model}. The objective of the system is not only to reliably deliver raw symbols, but to preserve the semantic meaning of the transmitted information under multi-constraints. In this model, the source of semantic information generates raw information that contains latent semantic aspects inherent in the data, such as contextual and relational meaning. The attention mechanism and the KG composed of entities and relations (detailed in section II-C) within the semantic encoder is then employed to identify and extract the most relevant semantic representations from this raw information by capturing long-range dependencies and contextual relationships. Rather than raw symbols, these extracted semantic representations form the input to the cGAN.
\begin{figure}
	\begin{center}
		\includegraphics[width=1.0\linewidth]{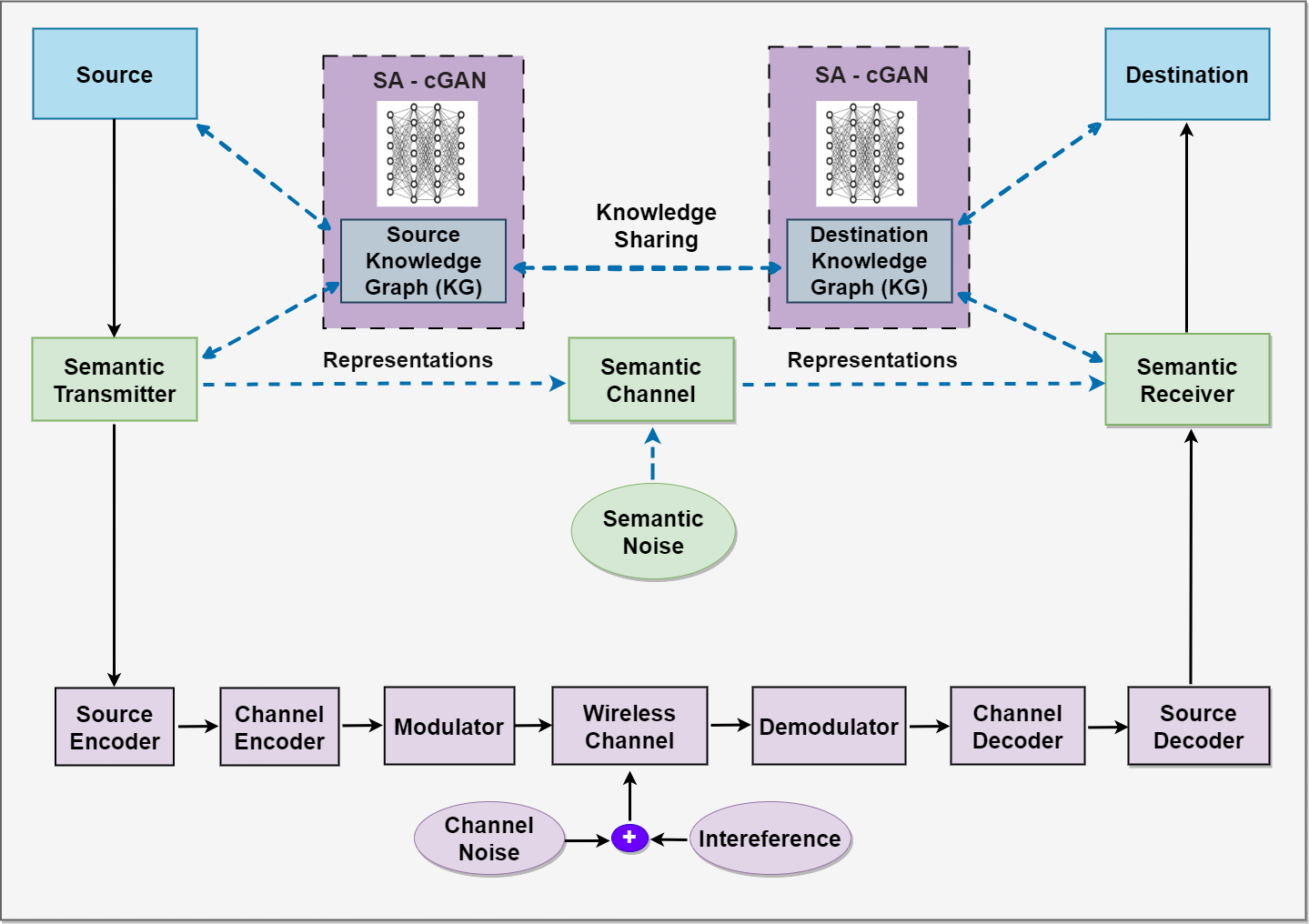}
        \vspace{-6mm}
		\caption[System model]{System model}
        \vspace{-6mm}
		\label{System model}
	\end{center}
\end{figure}
In the proposed framework, the cGAN learns a conditional mapping between semantic representations and compact symbolic embeddings that are optimized for transmission under bandwidth, distortion, and SNR constraints. RA is implicitly realized through this conditional generation process, where the cGAN adapts the number, structure, and fidelity of transmitted semantic symbols based on distortion criticality and channel conditions. As a result, essential semantic elements are generated and prioritized for transmission, enabling efficient bandwidth utilization while maintaining reliability and quality of service (QoS). As shown in~\eqref{1}, employing KG, the SA-cGAN assisted transmitter will encode the extracted semantics of the original information $(i)$ through the semantic encoder $\mathcal{H}_{s}(i)$. The output is then processed to perform channel encoding through the channel encoder $\mathcal{H}_{c}(\mathcal{H}_{s}(i))$. Subsequently, the SA-cGAN model $\mathcal{C}_{\text{SA--cGAN}}(\mathcal{H}_{c}(\mathcal{H}_{s}(i)))$ is used during the semantic and channel encoding for optimal performance, and transmits the encoded information $(I)$. 
\begin{equation} 
I_\text{encoded} = \mathcal{C}_{\text{SA--cGAN}}(\mathcal{H}_{c}(\mathcal{H}_{s}(i)))
\label{1}
\end{equation} 
However, as shown in~\eqref{2}, employing SA-cGAN with KG, the receiver will decode, extract the semantics, and restore the original information $(i)$ with reduced semantic and channel noise impacts on transmitted rates and QoS through dynamic RA and semantic representation adaptation.
\begin{equation}
i_\text{decoded} = \mathcal{C}_{\text{SA--cGAN}}(\mathcal{H}_{s}^{-1}(\mathcal{H}_{c}^{-1}(I)))
\label{2}
\end{equation} 
This semantic noise is considered a misinterpretation of semantics, which can be reduced by transmitting more semantic representation if needed to overcome the rate failures for increased rate resiliency. \emph{Rate failures} refer to instances where the effective transmission rate drops below the threshold required to preserve semantic fidelity, often due to channel impairments or semantic distortion. However, the \emph{rate resiliency} of the system is formally defined as its ability to consistently deliver seamless communication while maintaining a committed bit rate above a predefined threshold. \emph{Semantic adaptation} in the proposed framework is governed by an explicit trade-off between information representability and distortion criticality. \emph{Information representability} reflects the degree to which the selected semantic representations preserve the original meaning, while \emph{distortion criticality} captures the system’s tolerance to semantic degradation under bandwidth and channel constraints. When a channel is noise-free, the SA-cGAN optimizes the transmitted semantic representations to reduce bandwidth occupancy and energy consumption. Alternatively, when the channel is noisy, SA-cGAN performs the adaptation of the semantic representation and transmits more semantic elements that optimize semantic similarity, improving QoS against noise. To reduce the impact of noise, bandwidth utilization, and energy consumption, essential semantic elements will be transmitted, enticing the optimal resources, that will be sufficient to extract the actual meaning of the information at the receiver. The SA–cGAN adaptively navigates this trade-off by conditioning semantic generation on distortion criticality and network state, allowing flexible operation between conservative, moderate, and aggressive semantic adaptation approaches detailed in a subsequent section, depending on the desired balance between robustness and semantic fidelity. 
\vspace{-4mm}
\subsection{SA-cGAN Enabled 6G Semantic Communication}
We consider a network composed of multiple base stations (BS$_1$,BS$_2$,.....,BS$_n$) and mobile users (MU$_1$,MU$_2$,.....,MU$_n$), as shown in Fig.~\ref{SA-cGAN enabled 6G semantic communication}. The transmission medium is subject to various impairments, including channel noise and interference. In our framework, these diverse sources of degradation are collectively modeled as distortion, capturing their combined effect. When the accumulated distortion exceeds tolerable limits, the active link experiences rate degradation, and if the distortion becomes severe enough to prevent sustaining the committed rate, a rate failure event occurs. This event is detected by a proactive rate failure detection module, which signals the SA-cGAN controller to initiate a remediation and rate restoration phase. The SA-cGAN controller acts as a decision making entity that monitors channel quality (SNR), distortion criticality, and semantic reconstruction performance, and accordingly adjusts the semantic encoding strategy by selecting additional or refined semantic representations for transmission. Through this closed loop control mechanism, the SA-cGAN encoder continuously seeks optimal semantic representations that balance the trade-offs among SNR, committed bit rate, distortion criticality, and information representability, thereby ensuring resilient SemCom. As shown in Fig.~\ref{SA-cGAN enabled 6G semantic communication}, different strategies can be adopted: a conservative approach (high distortion and low information representability), a moderate approach (medium distortion and information representability), or an aggressive approach (low distortion and high information representability). Specifically, we can express the relationship as follows: If the information representability is low, this implies that the information representability loss (detailed in section III-B) is high. In such a case, the system should transmit a larger number of semantics, leading to a high semantic representation. Conversely, if the information representability is high, the corresponding information representability loss becomes low. In this case, fewer semantics need to be transmitted, resulting in a low semantic representation, lower bandwidth, and lower energy consumption. On the other hand, when distortion criticality is high, a larger number of semantic units must be transmitted, whereas fewer semantics are sufficient when distortion criticality is low. This trade-off can be intuitively explained using a teacher–student analogy. If a student already possesses the prerequisite knowledge in a course, the teacher can convey the intended concepts using fewer explanations, as the information representability is high and the information representability loss is low. In this case, even if minor distortions (unclear explanation) occur, the student can still infer the correct meaning, implying low distortion criticality and requiring fewer semantic elements to be transmitted. Conversely, if the student lacks prerequisite knowledge, information representability is low and the conveyed content becomes highly sensitive to distortion. Under such high distortion criticality, the teacher must provide more detailed explanations and repetitions to ensure correct understanding. Building on this principle, when improving the SNR through conventional means such as power control is impractical, the system compensates by employing semantic rate adaptation. In such cases, additional semantic features are transmitted to enrich the semantic representation and enhance robustness against noise. Although this strategy differs fundamentally from traditional communication paradigms, it is particularly effective in environments characterized by severe noise or unrecoverable channel impairments.
\begin{figure}
	\begin{center}
		\includegraphics[width=1\linewidth]{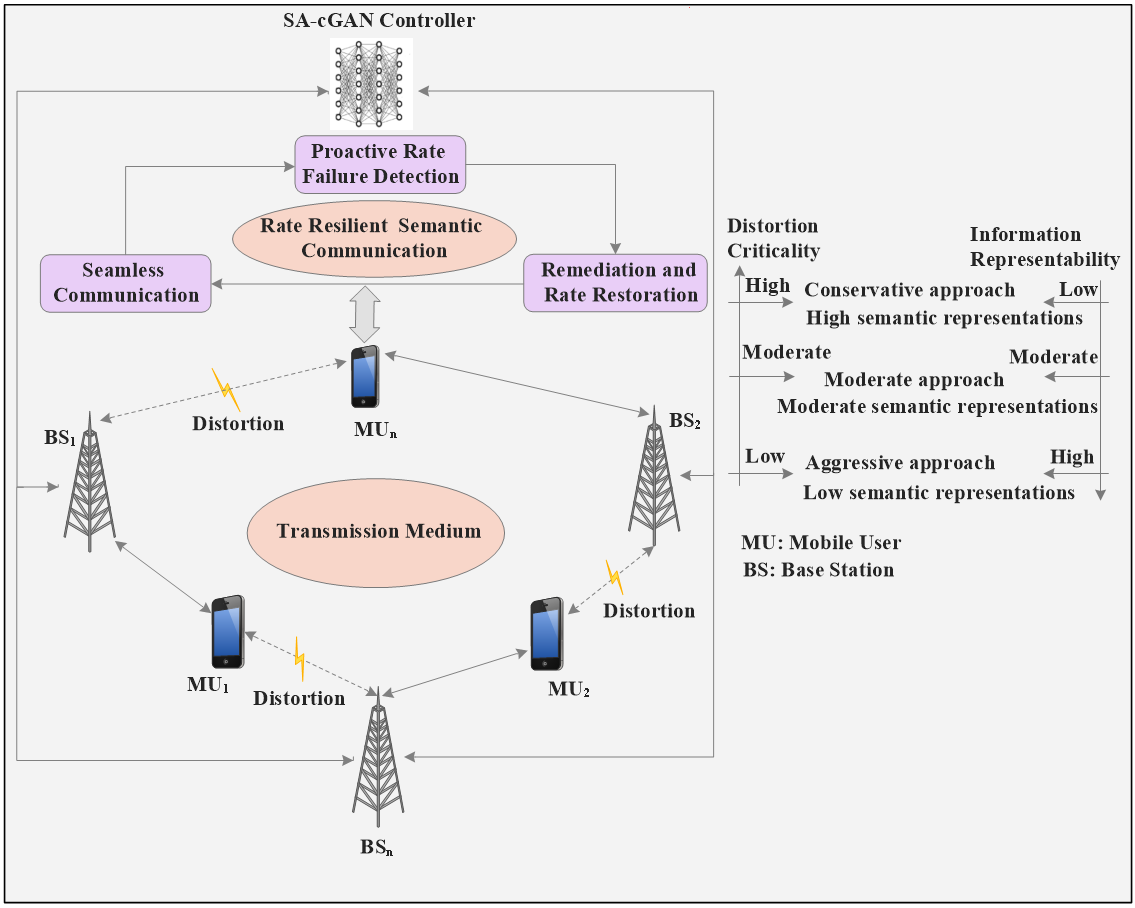}
        \vspace{-6mm}
		\caption[System model]{SA-cGAN enabled 6G semantic communication model}
        \vspace{-6mm}
		\label{SA-cGAN enabled 6G semantic communication}
	\end{center}
\end{figure}
\vspace{-4mm}
\subsection{Semantic Information Extraction}
A token $x_{i,n}$ is used to denote a word, symbol, or punctuation mark within textual data. Consequently, the content that the BS must deliver to user $i$ is represented as an ordered sequence of tokens, expressed in~\eqref{3}, as \cite{9832831};
\begin{equation}
X_i = \{x_{i,1}, x_{i,2}, \ldots, x_{i,n}, \ldots, x_{i,N_i}\}, \ \forall x_{i,n} \in V,
\label{3}
\end{equation}
where $V$ denotes the collection of tokens within a corpus, and $N_i$ indicates the total number of tokens in $X_i$. For instance, if the BS needs to send the sentence ``DNN is trained on large dataset." to user $i$, then, we obtain:
$X_i = \{[DNN], [is], [trained], [on], [large], [dataset], [.] \}$, where $x_{i,1} = [DNN], x_{i,2} = [is], x_{i,3} = [trained], x_{i,4} = [on], x_{i,5} = [large], x_{i,6} = [dataset],$ and $x_{i,7} = [.]$.

In our framework, the semantic information extracted from a textual data is modeled by a KG ($\mathcal{K}$) \cite{9416312}. Hence, the semantic information is formed by a collection of nodes along with a collection of edges, as shown in Fig.~\ref{Semantic extraction}. Each node within the semantic information corresponds to an entity that denotes an object or a concept in the real world. Hereafter, we define entity $j$ which is a subsequence of the text $X_i$ as $e_{i,j}$. For instance, in Fig.~\ref{Semantic extraction}, the phrase ``large dataset'' is treated as an entity composed of two tokens from the sample text. An information extraction system, such as the scientific information extractor described in \cite{luan-etal-2018-multi} can be employed to identify the set $\mathcal{E}_{i}$ of $E_i$ entities in text $X_i$.

Edges denote the relationships connecting pairs of entities. For any two identified entities $(e_{i,j}, e_{i,k})$, $j \neq k$, the BS must find the relation $r_{i,jk} \in \mathcal{R}_i$ between them, where $\mathcal{R}_i$ is the set of $R_i$ relations involved in text $X_i$. For example, in Fig.~\ref{Semantic extraction}, the relation between entity ``deep neural network'' and ``large dataset'' can be formulated as ``trained on''. It is important to note that the relationships (i.e., the edges of the semantic information) are directional, which means $r_{i,jk} \neq r_{i,kj}$. We assume that there is a predefined set $\mathcal{R}$ that includes all possible relations in the corpus, where each relation is expressed as a two-token sequence such as ``trained on'', ``used for'', ``part of'', and ``related to'', etc., as demonstrated in \cite{luan-etal-2018-multi}. The set $\mathcal{R}$ contains a finite number of relation types extracted from the corpus and knowledge base, i.e., $|\mathcal{R}| < \infty$, and its cardinality depends on the application domain and dataset. In our framework, $\mathcal{R}$ is constructed from the training corpus and associated knowledge base and includes commonly occurring semantic relations. The number of distinct relations is in the order of low hundreds. However, these relations are semantically clustered into higher-level predicate categories during training, which limits the effective classification complexity. Therefore, for any given pair of entities $e_{i,j}$ and $e_{i,k}$ in the original text $X_i$, the relation $r_{i,jk}$ between $e_{i,j}$ and $e_{i,k}$ can be obtained by classification algorithms such as convolutional neural networks \cite{lin-etal-2016-neural}. Here, the relation classification algorithm receives as input the sentence that includes the pair of entities from the original text, and its output is the predefined two-token relation summarizing the connection between them. We assume that the information extraction system is capable of detecting the majority of semantically relevant entities in the source text. Consequently, the triples structured as “entity–relation–entity” serve to capture and convey the meaning of the original text.

Using the identified entities together with the derived relations, the semantic information of the textual data $X_i$ can be represented as shown in~\eqref{4},\cite{9832831};
\begin{equation}
\mathcal{G}_i = \{\varepsilon_{i}^1, \ldots, \varepsilon_{i}^g, \ldots, \varepsilon_{i}^{G_i}\},
\label{4} 
\end{equation} \hfill
where, $\varepsilon_{i}^g = (e_{i,j}^g, r_{i,jk}^g, e_{i,k}^g), \ \forall e_{i,j}^g, e_{i,k}^g \in \mathcal{E}_i, \ j \neq k, \ \forall r_{i,jk}^g \in \mathcal{R}_i$ is a semantic triple and $G_i$ is the number of semantic triples in $\mathcal{G}_i$. Since each entity (e.g., $e_{i,j}^g$ or $e_{i,k}^g$) and each relation $r_{i,jk}^g$ consist of a sequence of tokens (e.g., $e_{i,j} = \{[deep], [neural], [network]\}$, $e_{i,k} = \{[large], [dataset]\}$, and $r_{i,jk} = \{[trained], [on]\}$), triple $\varepsilon_{i}^g$ can be expressed as $\varepsilon_{i}^g = \{v_{i,1}^g, \ldots, v_{i,b}^g, \ldots, v_{i,B_i^g}^g\}$, where $v_{i,b}^g \in V$ denotes the $b_{th}$ token within the semantic triple $\varepsilon_{i}^g$ and $B_i^g = S_{i,j}^g + S_{i,jk}^g + S_{i,k}^g,$ with $S_{i,j}^g$, $S_{i,k}^g$ and $S_{i,jk}^g$ represent the number of tokens in the entity $e_{i,j}^g$, $e_{i,k}^g$ and relation $r_{i,jk}^g$, respectively.
Then, the total number of tokens in the semantic information $\mathcal{G}_i$ is given by~\eqref{5},\cite{9832831};
\begin{equation}
Z(\mathcal{G}_i) = \sum_{g=1}^{G_i} \left(S_{i,j}^g + S_{i,k}^g + S_{i,jk}^g\right)
\label{5}
\end{equation}

As illustrated in Fig.~\ref{Semantic extraction}, the extracted semantic information occupies substantially less data space compared to the original text (i.e., $Z(\mathcal{G}_i) \ll N_i$). This is because, although semantic triples are constructed using tokens from the source text, $\mathcal{G}_i$ does not preserve all tokens in $X_i$. Instead, it is a selective semantic abstraction in which only entity and relation tokens are preserved. Tokens that do not contribute to semantic meaning, such as auxiliary verbs, determiners, and redundant descriptive phrases, are excluded. Consequently, even though semantic triples reuse tokens from the original vocabulary, the total number of transmitted tokens is significantly reduced. This reduction arises because the two-token relations that link the pairs of entities in $\mathcal{G}_i$ effectively eliminate the redundant context present in the original textual data $X_i$.

\begin{figure}
	\begin{center}
		\includegraphics[width=1.0\linewidth]{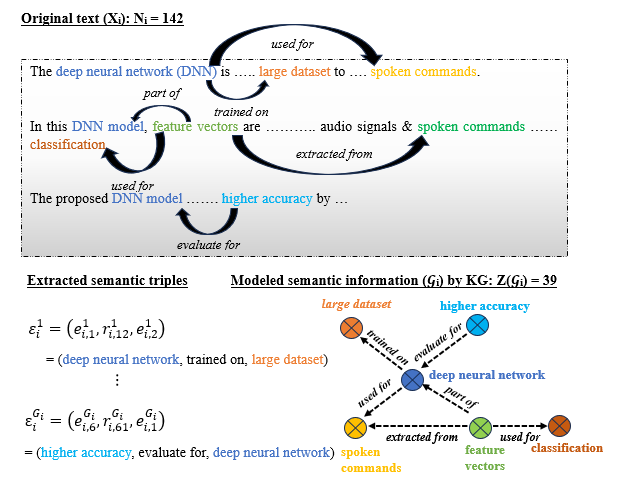}
        \vspace{-6mm}
		\caption{An illustration of source text along with the derived semantic information.}
        \vspace{-6mm}
		\label{Semantic extraction}
	\end{center}
\end{figure}
\vspace{-4mm}
\subsection{KG Guided Semantic Grounding and Token Selection}
In section II-B, we examined the process of interpreting the relationships between entities, which is used to build the KG as $\mathcal{K} = \{\mathcal{E}, \mathcal{R}, \mathcal{E}\}$, where $\mathcal{E}$, $\mathcal{R}$, and $\mathcal{E}$ are used to represent entities, relations, and entities, respectively. Next, in order to physically deliver the semantic information $\mathcal{G}_i$, it is necessary to form a sentence applying the defined semantic information extraction mechanism. As an illustration with an example of DNN training, the semantic relation learned by the system may be ``DNN, large dataset" connected by the relation/attribute ``trained on". The information we aim to deliver corresponds to the sentence: ``DNN is trained on large dataset". Specifically, semantic extraction was used to identify $\mathcal{E}$ and $\mathcal{R}$ from the input text and organize them into a KG representation. Although this symbolic structure captures high-level semantic relationships, it cannot be transmitted directly  over a physical communication channel. Therefore, a grounding mechanism is required to map symbolic semantic elements into neural representations that can be encoded, transmitted, and decoded reliably. In this work, KG guided semantic grounding refers to the process of anchoring symbolic elements derived from the $\mathcal{K}$ ($\mathcal{E}$ and $\mathcal{R}$) to their corresponding neural token embeddings, rather than grounding them in formal logical semantics. Specifically, grounding is performed during the semantic selection stage by explicitly aligning self-attention scores with KG derived semantic relevance. This ensures that tokens corresponding to meaningful semantic units, such as, entities and predicate relations are prioritized for transmission.

Let $X = \{x_1, x_2, \dots, x_n\}$ denote the tokenized input sentence, and let $\mathcal{K} = (\mathcal{E}, \mathcal{R})$ denote the KG constructed from $X$, where $\mathcal{E}$ is the set of extracted entities and $\mathcal{R}$ is the set of relations between them. Each token $x_i$ has a corresponding embedding vector $u_i$ and a self-attention/semantic importance score $\alpha_i$ (detailed in section III-B). Grounding is realized by jointly evaluating $\alpha_i$ and the token’s presence in $\mathcal{K}$. Formally, a token $x_i$ is considered grounded if it corresponds to either (i) an entity in $\mathcal{E}$ or (ii) a relation in $\mathcal{R}$. During semantic selection, such grounded tokens are assigned higher priority, while non-grounded tokens e.g., function words or syntactically necessary but semantically redundant tokens are deprioritized. This grounding-aware selection is implemented as a KG guided filtering and ranking operation applied to the self-attention scores (provided $\alpha_i \ge \tau_\alpha$) prior to semantic symbol generation. Concretely, grounding is implemented through a KG aware selection function that maps the tuple $(\alpha_i, x_i, \mathcal{K}) \;\longrightarrow\; s_i \in \{0,1\}$, where $s_i = 1$ (detailed in section III-B) indicates that token $x_i$ is selected for semantic transmission. This operation ensures that semantic selection is not solely driven by statistical attention, but is explicitly constrained by symbolic semantic structure encoded in the KG. The selected grounded embeddings $\{u_i \mid s_i = 1\}$ form the input to the SA--cGAN model, which learns compact semantic symbols for transmission. In this work, a semantic symbol refers to a compact latent representation generated by the cGAN from selected KG-grounded token embeddings. Unlike conventional modulation symbols that represent bit sequences, semantic symbols encode meaning-aware representations that are optimized with respect to distortion criticality and information representability prior to channel transmission. In this way, grounding acts as an interface between symbolic semantic representations and sub-symbolic neural encoding, ensuring that only semantically meaningful content is forwarded to the generative model.

It is important to emphasize that this grounding mechanism does not rely on formal logical inference, first-order logic, or fuzzy semantics as presented in \cite{BADREDDINE2022103649} and \cite{11076808}. Instead, grounding is achieved through KG-guided neural selection, where KG $\mathcal{E}$ and $\mathcal{R}$ provide semantic constraints that guide the attention-based $\alpha_i$ 
used for token selection. This approach remains computationally efficient and well aligned with the data-driven nature of the proposed SA–cGAN framework. By incorporating KG-based semantic structure during selection, the system prioritizes semantically meaningful tokens while reducing redundant information, thus improving semantic transmission efficiency.
\vspace{-4mm}
\subsection{Channel Model}
The robustness of the proposed SA-cGAN model is evaluated under Rayleigh fading channel. It is assumed that the transmission occurs in a dense environment consisting of multiple buildings and other surrounding obstacles. Hence, there may occur multiple scatterers resulting in a time-varying multipath fading channel which is assumed to be a Rayleigh fading channel. Let $r_k$ denote the semantic symbol transmitted at a time instance $k$, then the corresponding signal received at the receiver, as shown in~\eqref{6}, is given by 
\begin{equation}
y_k = h_k r_k + n_k,
\label{6}
\end{equation}
where $h_k \sim \mathcal{CN}(0,1)$ denotes the Rayleigh fading coefficient, and $n_k \sim \mathcal{CN}(0,\sigma_{n}^{2})$ represents additive white Gaussian noise (AWGN). For AWGN channels, $h_k$ = 1, whereas for Rayleigh channels it follows a fading distribution \cite{goldsmith2005wireless} and \cite{proakis2008digital}. This model captures both noise limited and fading dominated propagation conditions and is consistent with common SemCom modeling practices. Moreover, after transmission through the Rayleigh fading channel and noise corruption, the receiver performs channel decoding followed by semantic reconstruction to recover the textual output $\hat{X}_i$. The effectiveness of this end-to-end semantic transmission process is evaluated using the performance metrics, which are introduced next.
\vspace{-4mm}
\section{Performance Metrics and Problem Formulation}
\subsection{Performance Metrics}
In SemCom, successful transmission must be evaluated at multiple levels because preserving meaning is fundamentally different from merely reconstructing symbols. A single metric cannot fully characterize semantic fidelity (i.e., preservation of intended meaning) under bandwidth and channel constraints. Therefore, we evaluated the performance of the proposed method using four complementary metrics: \emph{bilingual evaluation understudy (BLEU) score}, \emph{semantic similarity}, \emph{semantic accuracy}, and \emph{semantic completeness} each capturing a distinct aspect of performance. Although some previous works loosely interpret BLEU as a measure of semantic accuracy \cite{9832831}, BLEU is in fact a syntactic n-gram overlap metric and is not equivalent to token-level semantic accuracy. Therefore, BLEU and semantic accuracy are related but not equivalent. However, BLEU does not capture semantic equivalence under paraphrasing or lexical substitution. For this reason, we additionally employ semantic similarity to measure sentence-level meaning preservation, semantic accuracy to quantify token-level correctness (precision), and semantic completeness to measure token-level coverage (recall). This multi-metric evaluation ensures that the proposed SA–cGAN framework is assessed not only for symbol reconstruction, but also for meaningful and contextually faithful semantic transmission. For clarity, we distinguish these metrics formally as follows.

\subsubsection{BLEU Score (Syntactic Fidelity)}
BLEU measures the preservation of the surface-level linguistic structure by counting $n$-gram matches between the decoded sentence $\hat{X}_i$ and the original sentence $X_i$. It is sensitive to word ordering and phrase continuity, and therefore reflects syntactic similarity rather than semantic meaning. BLEU is useful for evaluating symbol-level (i.e., token) reconstruction, but it cannot capture paraphrasing or meaning preservation. 
For the transmitted sentence $X_i$ with size $N_i$ and the decoded sentence $\hat{X}_i$ with size $\hat{N}_i$, the BLEU as shown in~\eqref{123}, can be expressed as
\begin{equation}
\log \mathrm{BLEU}_i
= \min\!\left(1 - \frac{\hat{N}_i}{N_i},\, 0\right)
+ \sum_{n=1}^{N} u_n \log p_n^{(i)},
\label{123}
\end{equation}

where $u_n$ is the weights of n-grams and $p_n^{(i)}$ is the n-grams score, defined in~\eqref{eq:p_n}, as 
\begin{equation}
p_n^{(i)} = \frac{\sum_{k} \min \left( C_k(\hat{X}_i), C_k(X_i) \right)}
{\sum_{k} \min \left( C_k(\hat{X}_i) \right)},
\label{eq:p_n}
\end{equation}
where $C_k(\cdot)$ denotes the count of occurrences of the $k$-th element in the $n$-gram sequence. The BLEU score evaluates the variation in n-grams between two sentences, where n-grams represent sequences of n consecutive words used for comparison. For example, for sentence ``This is my car,’’ 1-gram: ``This’’ ``is’’ ``my’’ and ``car’’ 2-grams: ``This is’’ ``is my’’ and ``my car’’. The same principle can be followed for the rest.

The BLEU score yields a value between 0 and 1, reflecting the degree of similarity between the decoded text and the original transmitted text, where a score of 1 represents perfect correspondence. However, few human translations will attain the score of 1, as minor word differences do not necessarily alter the overall meaning of a sentence. For example, the sentences “I bought a new car” and “I bought a new automobile” convey the same meaning but result in different BLEU scores due to lexical variation. To better capture such semantic equivalence, we introduce an additional metric semantic similarity, which operates at the sentence level to complement the BLEU score.

\subsubsection{ Semantic Similarity (Semantic-Level Meaning Preservation)}
A single word may convey different meanings depending on its context. For example, the word mouse carries distinct interpretations in the domains of biology and computing. Traditional embedding methods, such as \textit{word2vec} \cite{479394}, are unable to effectively capture such polysemy, as they represent each word with a fixed numerical vector that does not adapt to contextual variations. To capture this characteristic, semantic similarity is introduced as a new metric that measures how closely the meanings of two sentences align. Semantic similarity as shown in~\eqref{Sim}, computes the cosine similarity between contextual sentence embeddings e.g., bidirectional encoder representations from transformers  (BERT), i.e.,
\begin{equation}
\label{Sim}
\mathrm{Sim}_i = 
\frac{\mathbf{B}_\Phi(X_i)\cdot\mathbf{B}_\Phi(\hat{X}_i)^\top}
{\|\mathbf{B}_\Phi(X_i)\|\;\|\mathbf{B}_\Phi(\hat{X}_i)\|}
\end{equation}

where $\mathbf{B}_{\Phi}$ denotes the BERT model \cite{peters-etal-2018-deep}, which is a large-scale pre-trained network comprising billions of parameters, utilized for semantic feature extraction. The semantic similarity metric defined in~\eqref{Sim} yields a value between 0 and 1, indicating the degree of resemblance between the decoded sentence and the transmitted sentence, where 1 signifies complete similarity and 0 denotes no semantic correspondence between $X_i$ and $\hat{X_i}$. In contrast to BLEU, BERT has been trained on billions of sentences, enabling it to effectively learn and represent semantic information by generating context-specific vector representations. Additionally, unlike BLEU, this metric is invariant to synonyms, word order, or paraphrasing, and therefore reflects how well the meaning of the sentence survives the SA--cGAN semantic encoding, channel effects, and decoding.

\subsubsection{Semantic Accuracy (Token-Level Precision)}
Semantic accuracy as shown in~\eqref{eq:semantic_accuracy}, measures the proportion of decoded tokens that are correct with respect to the source sentence, i.e.,
\begin{equation}
\mathrm{Acc}_i = 
\frac{\sum_{x_{i,n}\in V} \min(\sigma(\hat{X}_i,x_{i,n}),\,\sigma(X_i,x_{i,n}))}
{\sum_{x_{i,n}\in V} \sigma(\hat{X}_i,x_{i,n})}
\label{eq:semantic_accuracy}
\end{equation}

where, $V$ denotes the global vocabulary or token set used for counting, which includes all words, subwords, and punctuation tokens defined by the tokenizer. Formally, $V = \{ x_{i,1}, x_{i,2},\dots \}$. Although the summation in~\eqref{eq:semantic_accuracy} extends over all $x_{i,n}\in V$, in practice it is restricted to tokens appearing in either $X_i$ or $\hat{X}_i$ to reduce computational complexity. Each $x_{i,n}\in V$ represents a single token (word, subword, or punctuation mark) from the vocabulary, and the summation accumulates token-level matches across the entire sentence. The function $\sigma(\hat{X}_i, x_{i,n})$ denotes the count of token  $x_{i,n}$ in the decoded (received) sentence $\hat{X}_i$, where $\sigma(\hat{X}_i,x_{i,n}) \in \{0,1,2,\dots\}$. For instance, if $\hat{X}_i = \{\text{the}, \text{cat}, \text{cat}\}$, then $\sigma(\hat{X}_i, \text{cat}) = 2$ and $\sigma(\hat{X}_i, \text{the}) = 1$. Similarly, $\sigma(X_i, x_{i,n})$ represents the count of token  $x_{i,n}$ in the original (transmitted) sentence $X_i$; for example, if $X_i = \{\text{the}, \text{cat}\}$, then $\sigma(X_i, \text{cat}) = 1$. The term $\min\!\big(\sigma(\hat{X}_i, x_{i,n}), \sigma(X_i, x_{i,n})\big)$ corresponds to the clipped (matched) count for token  $x_{i,n}$, which restricts the number of matches for that token to its occurrence frequency in the original message. This clipping prevents overcounting repeated tokens that may appear multiple times in the decoded output but fewer times in the source sentence, thereby ensuring that only legitimate token matches are counted in the accuracy computation. This metric evaluates whether the transmitted semantics introduce incorrect or hallucinated (information that the transmitter never sent) tokens. Unlike BLEU, it is order-independent and reflects token-level correctness. Thus, BLEU $\neq$ semantic accuracy: BLEU penalizes word-order and phrase mismatch, whereas semantic accuracy only measures correctness of recovered semantic elements.

\subsubsection{Semantic Completeness (Token-Level Recall)}
Semantic completeness as shown in~\eqref{eq:semantic_completeness}, measures how much of the original message’s content is present in the decoded sentence, i.e.,
\begin{equation}
\mathrm{Comp}_i = 
\frac{\sum_{x_{i,n}\in V} \min(\sigma(\hat{X}_i,x_{i,n}),\,\sigma(X_i,x_{i,n}))}
{\sum_{x_{i,n}\in V} \sigma(X_i,x_{i,n})} 
\label{eq:semantic_completeness}
\end{equation}

Using the same clipped token matching defined above, the numerator $\sum_{x_{i,n} \in V} \min\!\big(\sigma(\hat{X}_i, x_{i,n}), \sigma(X_i, x_{i,n})\big)$ represents the total number of correctly recovered tokens, and the denominator $\sum_{x_{i,n} \in V} \sigma(X_i, x_{i,n})$ denotes the total number of tokens present in the original message $X_i$. Hence, $\mathrm{Comp}_i$ quantifies the fraction of original semantic content that has been successfully reconstructed at the receiver, serving as a measure of semantic recall or completeness. A higher value of $\mathrm{Comp}_i$ indicates that the system has retained more of the transmitted meaning, whereas a lower value signifies semantic loss due to bandwidth compression, channel distortion, or cGAN reconstruction limitations. Semantic completeness therefore quantifies the information preservation achieved by semantic resource allocation.

These four metrics collectively capture the behavior of the proposed SA–cGAN model in managing the trade-off between distortion criticality and information representability.

\subsection{Problem Formulation}
Considering semantic transmission efficiency, we aim to design an optimization problem that selects which semantics to encode and transmit, under the trade-off between: (a) information representability loss $(\mathcal{L}_{\mathrm{IR}})$, that measures how well a few selected semantics can represent the full meaning of the original sentence, and (b) distortion criticality loss $(\mathcal{L}_{\mathrm{DC}})$, that measures how vulnerable the selected semantics is to distortion due to channel and symbolic ambiguity, especially under fading channel and imperfect cGAN reconstruction. 

Let $X = \{ x_1, x_2, \dots, x_n \}$ denote the input sentence (or source message) composed of $n$ tokens, and let $S = \{ s_1, s_2, \dots, s_k \} \subseteq X$ denote the set of semantic tokens (i.e., semantics) selected from $X$ based on the semantic importance scores $\alpha_i$ using the self-attention mechanism. The semantic importance scores $\alpha_i$ are computed via a self-attention mechanism over token embeddings $\{u_i\}_{i=1}^{n}$, where the learnable parameters assign higher weights to tokens that contribute more significantly to the overall semantic embedding $\Phi(X)$. These scores are further refined using KG guided grounding to prioritize tokens corresponding to extracted $\mathcal{E}$ and $\mathcal{R}$. We define binary semantic selection variables: $s_i \in \{0,1\}, \quad \forall i \in \{1, \dots, |X|\}$, where $s_i = 1$ indicates that token $x_i$ is selected in $S$, provided $\alpha_i \ge \tau_\alpha$. On the other hand, if $\alpha_i < \tau_\alpha$, then  $s_i = 0$ (i.e., dropped or compressed). Here, $\alpha_i$ denotes the self-attention score/weight and $\tau_\alpha$ is a predefined threshold for importance that can be tuned based on design needs, and $|X|$ indicates the total number of semantic units in the input sentence X. We also define $\Phi(X)$ as semantic embedding of the original message $X$ which is equal to $\sum_{i=1}^n\alpha_i u_i$. $\Psi(S)$ as an aggregated semantic embedding of the selected subset $S$ which is equal to $\sum_{i=1}^n s_i u_i$. Where, $u_i$ denotes the embedding vector for token $x_i$ (or $s_i$ when selected (i.e., $\Phi(s_i)$). \(c_i: \beta_i \Big[ \mathbb{E}_{z \sim \mathcal{N}(0,1)}\! \| x_i - D(G(z|x_i))\|_2^2 + \kappa \, \mathrm{SNR}(x_i)^{-1} \Big]\) represents the distortion–reliability penalty associated with the transmission of the semantic token $x_i$ (provided $\forall i : s_i=1$ i.e., $x_i \in S$), reflecting its semantic reconstruction difficulty and channel unreliability.
Here, $G(\cdot)$ and $D(\cdot)$ denote the generator and decoder (discriminator-assisted reconstruction module) of the cGAN used for the generation and reconstruction of semantic symbols, respectively. Specifically, $G(z|x_i)$ represents the semantic symbol generated by the cGAN conditioned on the semantic token $x_i$ and a latent random vector $z$, while $D(\cdot)$ reconstructs the semantic representation from the generated symbol. The latent variable $z \sim \mathcal{N}(0,1)$ denotes the stochastic input to the generator, allowing the model to learn a distribution of semantic representations conditioned on the input semantic tokens. $\beta_i$ denotes the weight of the token $x_i$ in the channel cost. In our framework, cost means the transmission burden and the reliability penalty of sending a semantic unit. It is a semantic physical penalty term that is used to guide optimal token selection. The higher $\beta_i$ indicates that the token $x_i$ is more critical for transmission reliably. In addition, $k$ is a bounding parameter associated with the inverse term of the SNR. The higher $k$ indicates that the optimization penalizes the tokens more if they have a low SNR. Following the above discussion, the mathematical formulations of $\mathcal{L}_{\mathrm{IR}}$ and $\mathcal{L}_{\mathrm{DC}}$ are provided in \eqref{IR} and \eqref{DC}, respectively. 

\begin{equation}
\label{IR}
    \mathcal{L}_{\mathrm{IR}}= \left\| \Phi(X) - \Psi(S) \right\|_2^2 
\end{equation}

\begin{equation}
\begin{aligned}
\mathcal{L}_{\mathrm{DC}} 
= \sum_{i=1}^n \beta_i \Big(
& \mathbb{E}_{z \sim \mathcal{N}(0,1)} 
\left\| x_i - D(G(z|x_i)) \right\|_2^2 \\
& + \kappa \, \mathrm{SNR}(x_i)^{-1}
\Big)
\end{aligned}
\label{DC}
\end{equation}

Together, $\mathcal{L}_{\mathrm{IR}}$ and $\mathcal{L}_{\mathrm{DC}}$ quantify the semantic representability distortion trade-off that governs optimal token selection. The objective is therefore to identify a subset $S \subseteq X$ that minimizes semantic embedding deviation while controlling semantic reconstruction difficulty and channel unreliability. Building upon these loss components, we now formulate a semantic-adaptive extended optimization problem that explicitly incorporates practical transmission constraints. Algorithm \ref{Algo 1} presents the SA-cGAN enabled 6G SemCom workflow module.

\begin{algorithm}
\caption{SA-cGAN enabled 6G semantic communication workflow}
\label{Algo 1}
\begin{algorithmic}[1]
\STATE \textbf{Input:} $X = \{x_1, x_2, \dots, x_n\}$
\STATE \textbf{Output:} Decoded sentence $\hat{X}$ and metrics (BLEU, $\mathrm{Sim}_i$, $\mathrm{Acc}_i$, $\mathrm{Comp}_i$)
\STATE Apply self-attention to obtain importance scores $\alpha_i$
\STATE Select subset $S \subseteq X$ based on $\alpha_i$ and $\tau_{\alpha}$
\STATE Refine $S$ using KG relations for semantic consistency
\STATE Generate semantic symbols via cGAN generator: $G(z|x_i)$, $\forall x_i \in S$
\STATE Transmit symbols through THz channel with noise and fading
\STATE Receive distorted symbols and decode via cGAN decoder $D(.)$
\STATE Reconstruct semantics $\hat{S}$ using KG
\STATE Aggregate reconstructed semantics $\Psi(\hat{S})$
\STATE Compare $\hat{S}$ with original $X$ via $\Phi(X)$
\STATE Compute BLEU, $\mathrm{Sim}_i$, $\mathrm{Acc}_i$, and $\mathrm{Comp}_i$
\IF{rate failure or semantic loss detected}
    \STATE Adapt RA and retransmit additional semantic representations
\ENDIF
\STATE Output $\hat{X}$ and final performance metrics
\end{algorithmic}
\end{algorithm}

\subsection*{Semantic Adaptive Extended Optimization Problem}
Motivated by the objective of meaning centric communication, we formulate a semantic adaptive optimization problem that explicitly accounts for semantic fidelity, distortion criticality, and information representability . Unlike traditional formulations that optimize raw transmission rates or symbol error metrics, the proposed problem seeks to identify an optimal subset of semantic representations that preserves the intended meaning while efficiently utilizing limited communication resources. Specifically, the formulation jointly balances information representability, distortion criticality, bandwidth occupancy, and channel reliability, thereby capturing the fundamental trade-offs inherent in SemCom. Based on the study of problem formulation, we present an extended optimization problem, as shown in~\eqref{Ext Opt.}, that rigorously models semantic transmission efficiency under practical constraints.
\begin{equation}
\label{Ext Opt.}
\min_{\bm{s} \in \{0,1\}^n} \&\quad 
\lambda_1 \left\| \Phi(X) - \Psi(S) \right\|_2^2
+ \lambda_2 \sum_{i=1}^n c_i s_i \\
\end{equation}
\begin{subequations}
\begin{align}
\text{s.t.} \quad 
& C_1: \sum_{i=1}^n s_i \le k_{\max}, \tag{14a}
\\
& C_2: \alpha_i \ge \tau_\alpha, \forall i : s_i=1, \tag{14b}
\\
& C_3: \mathbb{E}_{z \sim \mathcal{N}(0,1)}\!\left[\| x_i - D(G(z|x_i)) \|_2^2 \right] \le \epsilon, \forall i : s_i=1, \tag{14c}
\\
& C_4: KL(P_{\mathrm{src}}(x_i) \| P_{\mathrm{dest}}(x_i)) \le \delta, \forall i : s_i=1, \tag{14d}
\\
& C_5: \mathrm{Sim}(x_i, x_j) \le \tau_{\mathrm{sim}}, \forall i\neq j, : s_i = s_j = 1, \tag{14e}
\\
& C_6: \mathrm{SNR}(x_i) \ge \gamma_{\min}, \forall i : s_i=1 \tag{14f}
\end{align}
\end{subequations}
where (14a) is a bandwidth constraint that indicates the maximum number of semantics to be transmitted within the bandwidth budget, and (14b) represents a semantic importance threshold constraint that ensures only tokens with self-attention weights above a predefined threshold are selected. The constraint (14c) defines the cGAN reconstruction constraint, which indicates that the semantic symbols generated from cGAN must be close to the original semantics. Further, (14d) introduces a KG distribution alignment constraint. Here, $P_{\mathrm{src}}$ and $P_{\mathrm{dest}}$ denote probability distributions over the set of semantic $\mathcal{R}$ (or semantic triples) derived from the source and reconstructed KG's, respectively. These distributions represent the relative frequencies of semantic $\mathcal{R}$ appearing in the message-level KG. Due to semantic compression, channel distortion, and imperfect cGAN reconstruction, the set of recovered semantic $\mathcal{R}$ at the receiver may differ slightly from those extracted at the transmitter. The Kullback-Leibler (KL) divergence term therefore measures the discrepancy between the semantic relational structures of the two KG's. By constraining this divergence below a threshold, the proposed framework ensures that the overall relational semantics of the message remain consistent, even if some tokens are omitted or paraphrased during transmission. This formulation integrates KG aware semantic structure directly into the optimization problem, providing a principled mechanism to preserve semantic integrity beyond symbol level reconstruction and constituting an additional novelty aligned with the knowledge based SemCom paradigm. \emph{Semantic integrity} refers to the preservation of the core $\mathcal{E}$--$\mathcal{R}$ structure and contextual meaning of the original message after semantic selection, compression, transmission, and reconstruction, regardless of exact lexical or syntactic variations. The constraint (14e) describes the information redundancy constraint, where $\mathrm{Sim}(.,.)$ is a semantic similarity metric denotes the cosine similarity between contextual BERT embeddings that prevents redundant semantic transmission and promotes diversity in selected symbols. The constraint (14f) specifies the transmission reliability, which indicates that each semantic symbol must meet a minimum SNR for reliable transmission. Additionally, the term $s_i$ is multiplied with the  $\mathcal{L}_{\mathrm{DC}}$ in~\eqref{Ext Opt.}, to count the cost only for symbols that are selected, whereas $\lambda_1$ and $\lambda_2$ are weight parameters for $\mathcal{L}_{\mathrm{IR}}$ and $\mathcal{L}_{\mathrm{DC}}$, respectively.
\vspace{-4mm}
\section{Solution for Optimization Problem}
Given the complexity of the multi-constraint semantic selection problem, an exact combinatorial solution is infeasible within practical computational limits. Therefore, a complexity analysis is a must to show that the non-deterministic polynomial (NP) hard optimization problem when presented as a heuristic generates approximate results but is computationally tractable. Hence, we construct a systematic solution approach that first relaxes the discrete problem into a convex form and then leverages submodularity to enable efficient approximation. Subsequently, a Lagrangian dual framework is developed to manage constraint interactions, followed by Karush–Kuhn–Tucker (KKT) conditions to characterize optimality. The detailed steps are presented below.
\vspace{-2mm}
\subsection{Continuous Relaxation}
The original problem involves binary variables (select or not select semantic tokens), making it a combinatorial and non-convex optimization that is NP hard in general. Thus, we are relaxing the binary constraints to allow continuous values (between 0 and 1) which transform the problem into a convex or more tractable form. This enables the use of efficient convex optimization techniques that can find tractable solutions or lower bounds, which is time consuming and complex in the original discrete formulation. Relaxing the binary variables implies that
\begin{equation}
0 \le s_i \le b_i, \quad b_i = {1}_{\{\alpha_i \ge \tau_\alpha\}}    
\end{equation}

For each subset S, we aggregate: $\Psi(S) = \sum_{i=1}^n s_i u_i$ (linear aggregation):
\begin{equation}
\label{Cont. Rel}
\min_{{s} \in \mathbb{R}^n}
\lambda_1 \left\|\Phi(X) - \sum_i s_i u_i \right\|_2^2
+ \lambda_2 \sum_i c_i s_i\quad
\end{equation} 
\text{s.t. constraints as in~\eqref{Ext Opt.} hold with } $s_i\in [0,b_i]$.

The relaxed forms of new constraints are given as:
\begin{equation}
 \mathbb{E}_{z}\left[\| x_i - D(G(z|x_i)) \|_2^2 \right] s_i \le \epsilon, \quad s_i \in [0,b_i],   
\end{equation}
\begin{equation}
  s_i \times \mathrm{SNR}(x_i) \ge \gamma_{\min} s_i, \quad s_i \in [0,b_i]  
\end{equation}
The above relaxation effectively enforces reconstruction and SNR bounds on selected tokens, and sets the foundation for further analytical and algorithmic advances.

\begin{lemma}[Convexity verification (relaxed selection)]
Consider the relaxed problem with $0 \le s_i \le b_i$, 
where $b_i = \mathbf{1}\{\alpha_i \ge \tau_\alpha\}$, 
and linear aggregation 
$\Psi(S) = \sum_{i=1}^n s_i u_i = U\mathbf{s}$,
with $U = [u_1,\dots,u_n] \in \mathbb{R}^{d \times n}$, denotes the matrix of embeddings $u_i$.

The relaxed optimization problem is given by
\begin{equation}
\begin{aligned}
\min_{\mathbf{s} \in \mathbb{R}^n} \quad 
& \hat{F}(\mathbf{s}) 
= \lambda_1 \|\Phi(X) - U\mathbf{s}\|_2^2 
+ \lambda_2 \sum_{i=1}^n c_i s_i \\
\text{s.t.} \quad 
& 0 \le s_i \le b_i, \quad \forall i,
\end{aligned}
\label{Convexity Verification}
\end{equation}
together with the relaxed constraints defined in~\eqref{Ext Opt.}. If any additional surrogate terms (e.g., convex BLEU surrogates) are convex, then the full relaxed objective with the listed constraints constitutes a convex optimization program.
\end{lemma}

\begin{proof}
The quadratic $\mathcal{L}_{\mathrm{IR}}$ term in \eqref{Convexity Verification} expands as $\|\Phi-U\mathbf{s}\|_2^2=\|\Phi\|_2^2-2\Phi^\top U\mathbf{s}+\mathbf{s}^\top(U^\top U)\mathbf{s}$, so its Hessian is $\nabla^2(\lambda_1\|\Phi-U\mathbf{s}\|_2^2)=2\lambda_1 U^\top U$. Since $U^{\top}U$ is positive semidefinite and $\lambda_1 \ge 0$, hence this term is convex. Further, the channel-cost term $\lambda_2\sum_i c_i s_i$ and the cGAN linearized constraint $C_3\le\epsilon$ are affine in $\mathbf{s}$ and thus convex. Also, the box constraints $0\le s_i\le b_i$ and other budget/linear constraints are convex sets. If additional non-convex evaluation metrics (e.g., BLEU) are incorporated into the optimization, they can be replaced by convex surrogate terms to maintain tractability. That means, any additional surrogate is assumed convex, resulting in a convex contribution. Therefore, the objective in \eqref{Ext Opt.} is a sum of convex functions, and the feasible set is an intersection of convex sets, implying that the relaxed problem is a convex program. Moreover, if $\lambda_1>0$ and $U^\top U\succ0$ on the active columns then the Hessian is positive definite and the objective is strictly convex, yielding a unique minimizer. 
\end{proof}
\vspace{-5mm}
\subsection{Submodularity and Approximability}
The objective function in \eqref{Ext Opt.} with its semantic representation error and additive cost terms, exhibits submodular properties—diminishing returns and modular additive parts. Under mild assumptions on the embedding geometry (e.g., nonnegative pairwise similarities $u_i^\top u_j \ge 0$), the resulting set function exhibits the diminishing returns property and is therefore submodular. This structure allows the use of greedy algorithms with theoretical approximation guarantees, specifically a $(1-1/e)$ approximation ratio, which is a strong theoretical guarantee for near-optimality, where $e$ represents the base of the natural logarithm. In other words, exploiting submodularity significantly reduces computational complexity and provides provable near-optimality for discrete selection problems under combinatorial constraints like matroids (e.g., semantic diversity) requirements and knapsacks (e.g., bandwidth limit). The bandwidth constraint (14a) imposes a cardinality/knapsack-type restriction on the total number of selected semantic tokens. The diversity constraint (14e), which prevents simultaneous selection of highly similar semantic tokens, induces an independence structure that can be modeled as a matroid constraint. The feasible region therefore corresponds to the intersection of a matroid and a knapsack constraint. This property guides the design of practical algorithms that are efficient, effective, and ensuring suitability for large scale semantic selection problems in 6G systems. 

\begin{lemma}
If $\Psi(S)=\sum_{i\in S} u_i$ is linear and pairwise similarities satisfy $u_i^\top u_j \ge 0$, then the set function
\begin{equation}
F(S) = \lambda_1 \|\Phi(X) - \Psi(S)\|_2^2 + \lambda_2 \sum_{i\in S} c_i
\end{equation}
is submodular. Consequently, maximizing $-F(S)$ under cardinality, matroid, or knapsack constraints provides a $(1 - 1/e)$ approximation guarantee via greedy or continuous greedy algorithms.
\end{lemma}

\begin{proof}
We analyze each component of \eqref{Ext Opt.} as,
(i) Quadratic semantic representation error: with $\Psi(S)=\sum_{i\in S} u_i$, expansion gives $\|\Phi(X)\|^2 - 2\Phi(X)^\top \sum_{i\in S} u_i + \|\sum_{i\in S} u_i\|^2$. The corresponding marginal gain for an element $j$ is $\Delta_j(A) = 2\Phi(X)^\top u_j - 2(\sum_{i\in A} u_i)^\top u_j - \|u_j\|^2$. For $S \subseteq T$, $\Delta_j(S)-\Delta_j(T) = 2(\sum_{i\in T\setminus S} u_i)^\top u_j \ge 0$, proving diminishing returns and hence submodularity. (ii) Additive cost: $\sum_{i\in S} c_i$ is modular (linear in $s_i$), which preserves submodularity under addition, which means that the contribution of each semantic token is independent of the selection of other tokens. (iii) Combining terms: Since the sum of a submodular function and a modular function remains submodular, the overall objective $F(S)$ is therefore submodular. Its negation $-F(S)$ is monotone submodular under nonnegative weights, and under a cardinality or matroid/knapsack constraint, \cite{article} guarantees a $(1-1/e)$ approximation via greedy or continuous greedy algorithms. Here, the knapsack constraint captures the finite semantic transmission budget, while the matroid structure arises from semantic diversity and dependency constraints that restrict the joint selection of highly similar semantic tokens, ensuring coverage and redundancy control. 
\end{proof}

\begin{lemma}[Matroid intersection lemma (feasible augmentation under matroid intersection)]
Let $\mathcal{M}_1=(E,\mathcal{I}_1)$ and $\mathcal{M}_2=(E,\mathcal{I}_2)$ be two matroids on the same ground set $E=\{1,\dots,n\}$. If $A\subseteq E$ is independent in both matroids ($A\in\mathcal{I}_1\cap\mathcal{I}_2$) and $B\in\mathcal{I}_1$ with $|B|>|A|$, then there exists $b\in B\setminus A$ such that $A\cup\{b\}\in\mathcal{I}_1\cap\mathcal{I}_2$ or there exists $a\in A\setminus B$ with $B\cup\{a\}\in\mathcal{I}_1\cap\mathcal{I}_2$.
\end{lemma}

\begin{proof}
This is a standard augmentation property consequence of the matroid exchange axiom applied to intersection structures as detailed by the Edmonds’ matroid intersection theorem \cite{Edmonds2003}. An exchange is first applied on $\mathcal{M}_1$, followed by a repair step on $\mathcal{M}_2$ using alternating augmenting paths in the auxiliary exchange graph. This matroid intersection lemma guarantees that, in formulated semantic selection problem, it is always possible to safely add or swap semantic tokens while respecting both diversity (14e)  constraints and semantic structure (14d) constraints, ensuring that greedy or augmentation algorithms progress correctly toward a near-optimal semantic subset.
\end{proof}

In the proposed formulation, the ground set $E$ corresponds to candidate semantic tokens, while $\mathcal{M}_1$ models diversity constraints (14e) that prevent redundant semantic selection, and $\mathcal{M}_2$ models structural or KG-based semantic integrity constraints (14d). This lemma guarantees that if a larger feasible semantic subset exists, then it is always possible to augment or exchange elements while maintaining feasibility under both constraints. This property ensures that greedy or augmentation-based algorithms can iteratively improve the semantic subset without being stuck in infeasible configurations. Consequently, the matroid intersection structure provides a theoretical foundation for the correctness and convergence of the proposed constrained semantic selection strategy.
\begin{algorithm}
\caption{Greedy semantic selection with multi-constraints}
\label{Algo 2}
\begin{algorithmic}[1]
\STATE \textbf{Word and token extraction:} 
Given input sentence $X = \{x_1, x_2, \dots, x_n\}$, decompose into tokens with embeddings $u_i$ and importance scores $\alpha_i$
\STATE \textbf{Semantic embedding setup:} 
Compute full semantic embedding $\Phi(X)$, and define the aggregated embedding for any selected subset $S$ as
$\Psi(S) = \sum_{i \in S} u_i$
\STATE \textbf{Binary variables:} 
Associate each token $x_i$ with a binary selection variable $s_i \in \{0,1\}$, where $s_i=1$ if $x_i$ is selected
\STATE \textbf{Information representability loss $(\mathcal{L}_{\mathrm{IR}})$:} 
Compute the $(\mathcal{L}_{\mathrm{IR}})$
\[
\mathcal{L}_{\mathrm{IR}}= \left\| \Phi(X) - \Psi(S) \right\|_2^2
\]
if the information representability is low, $\mathcal{L}_{\mathrm{IR}}$ is high
\STATE \textbf{Distortion criticality Loss $(\mathcal{L}_{\mathrm{DC}})$:} 
Compute the $(\mathcal{L}_{\mathrm{DC}})$
\[
    \mathcal{L}_{\mathrm{DC}}= \sum_{i=1}^n\Big[\beta_i \Big[ \mathbb{E}_{z} \| x_i - D(G(z|x_i))\|_2^2 + \kappa \, \mathrm{SNR}(x_i)^{-1} \Big]\Big]
\]
lower is better.
\STATE \textbf{Candidate filtering:}
Define the candidate set
\[
\mathcal{C} \gets \left\{ i \;\middle|\;
\begin{aligned}
\alpha_i &\ge \tau_\alpha, \\
\mathbb{E}_{z}\!\left[\| x_i - D(G(z|x_i))\|_2^2 \right] &\le \epsilon, \\
\mathrm{SNR}(x_i) &\ge \gamma_{\min}, \\
KL_i &\le \delta
\end{aligned}
\right\}
\]

\STATE Initialize selected set $S \gets \emptyset$

\WHILE{$|S| < k_{\max}$ \textbf{and} $\mathcal{C} \neq \emptyset$}
    \FOR{each $i \in \mathcal{C}$}
        \STATE \textbf{Check diversity constraint:} ensure 
        $\mathrm{Sim}(x_i,x_j) \le \tau_{\mathrm{sim}}$ for all $j \in S$; 
        if violated, mark $i$ as infeasible and skip

        \STATE \textbf{Compute marginal gain:}
        \[
         \begin{aligned}
              \Delta_i =
              & \lambda_1 \Big( \|\Phi(X)-\Psi(S)\|_2^2 \\
              & - \|\Phi(X)-\Psi(S \cup \{i\})\|_2^2 \Big) \\
              & - \lambda_2 c_i
              \end{aligned}
        \]
    \ENDFOR
    \STATE Select $i^* = \arg\max_{i \in \mathcal{C}} \Delta_i$ among feasible candidates. The greedy step maximizes $\Delta_i$, which represents the net weighted reduction in the $(\mathcal{L}_{\mathrm{IR}})$ obtained by adding token $x_i$, penalized by $c_i$.
    \STATE Update $S \gets S \cup \{ i^* \}$ and $\mathcal{C} \gets \mathcal{C} \setminus \{ i^* \}$
\ENDWHILE
\STATE \textbf{return} $S$
\end{algorithmic}
\end{algorithm}

\begin{theorem}[Greedy guarantee with matroid + knapsack constraints]
Suppose $g(S)$ is nonnegative monotone submodular and let the feasible family be $\mathcal{F}=\mathcal{I}\cap\{S:\sum_i w_i s_i \le W\}$ where $\mathcal{I}$ is a matroid family and $w_i\ge0$. Then the continuous greedy + swap rounding scheme yields a $(1-1/e-\varepsilon)$ approximation to $\max_{S\in\mathcal{F}} g(S)$ for any $\varepsilon>0$, in polynomial time (polynomial in $n$ and $1/\varepsilon$).
\end{theorem}

\begin{proof}
The continuous greedy algorithm of Vondrák \cite{10.1145/1374376.1374389} 
constructs a fractional solution $x$ for the multilinear extension of $g$ that achieves a $(1-1/e)$ approximation factor. The swap rounding scheme of \cite{5671314} then converts $x$ into an integral solution $S \in \mathcal{F}$ while preserving the feasibility under the constraints of the matroid and the knapsack while incurring only a small loss, $\varepsilon$, which arises from discretization. Hence, the overall approximation factor is $(1-1/e-\varepsilon)$ following the classical analyses in the submodular optimization literature\cite{article}.
\end{proof}

This theorem guarantees that a hard combinatorial semantic-selection problem with matroid constraints (semantic structure, diversity) and knapsack constraints (bandwidth) can be solved efficiently, achieving at least 63\% of the optimal value via continuous greedy and swap rounding—the best possible polynomial-time guarantee for submodular objectives. In our case, the feasible set is the intersection of matroid constraints ($C_4$ and $C_5$) with a knapsack constraint ($C_1$). This exactly matches the theoretical framework that allows a $(1-1/e-\varepsilon)$  approximation using continuous greedy and swap rounding.

The algorithm \ref{Algo 2} is the classical greedy algorithm for submodular minimization/maximization, and it has been proven earlier that our function is submodular. At every step, the algorithm selects the token whose addition reduces semantic distortion the most, while accounting for its transmission cost.

\begin{algorithm}
\caption{Lagrangian dual ascent for relaxed semantic selection problem}
\label{Algo 3}
\begin{algorithmic}[1]
\STATE Initialize multipliers $\lambda, \mu_i, \nu_{ij}, \eta_i, \theta_i \gets 0$
\WHILE{not converged}
    \STATE Solve convex quadratic program (QP) for ${s}$ given current multipliers
    \STATE Update dual variables:
    \STATE $\lambda \gets [\lambda + \eta (\sum_i s_i - k_{\max})]_+$
    \STATE $\mu_i \gets [\mu_i + \eta (KL_i - \delta)]_+, \quad \forall i$
    \STATE $\nu_{ij} \gets [\nu_{ij} + \eta (s_i+s_j - (1+\tau'_{ij}))]_+, \quad \forall i<j$
    \STATE $\eta_i \gets [\eta_i + \eta (\mathbb{E}_{z}\left[\| x_i - D(G(z|x_i))\|_2^2 \right] s_i - \epsilon)]_+, \quad \forall i$
    \STATE $\theta_i \gets \left[\theta_i + \eta \left(s_i(\gamma_{\min} - \mathrm{SNR}(x_i))\right)\right]_+, \quad \forall i$
\ENDWHILE
\STATE Threshold/round ${s}$ to obtain discrete set $S$
\end{algorithmic}
\end{algorithm}
\vspace{-4mm}
\subsection{Lagrangian Dual Formulation}
Many constraints in \eqref{Ext Opt.} are complex and coupled, making direct constrained optimization difficult. Formulating the Lagrangian dual problem converts constrained optimization into an unconstrained dual maximization problem involving Lagrange multipliers \cite{boyd2004convex}. This enables iterative dual ascent methods, where primal variables and multipliers are updated alternately, simplifying constraint handling.

The Lagrangian for the relaxed problem as expressed in~\eqref{Cont. Rel}, including SNR and other constraints, is shown in~\eqref{Lang}

\begin{equation}
\begin{split}
\mathcal{L}(s, \lambda, \mu, \nu, \eta, \theta)  =\;
& \lambda_1 \Big\|
\Phi(X) - \sum_{i=1}^n s_i u_i
\Big\|_2^2 \\
& + \lambda_2 \sum_{i=1}^n c_i s_i \\
& + \lambda \left(\sum_{i=1}^n s_i - k_{\max}\right) \\
& + \sum_{i=1}^n \mu_i KL_i
- \sum_{i=1}^n \mu_i \delta \\
& + \sum_{i<j} \nu_{ij} (s_i + s_j - 1 - \tau'_{ij}) \\
& + \sum_{i=1}^n \eta_i \Big(
\mathbb{E}_z [ \|x_i - D(G(z|x_i))\|_2^2 ]
- \epsilon \Big) \\ s_i
& + \sum_{i=1}^n \theta_i
( \gamma_{\min} - \mathrm{SNR}(x_i) ) s_i
\end{split}
\label{Lang}
\end{equation}
where, Lagrange multipliers $\lambda, \mu_i, \nu_{ij}, \eta_i, \theta_i \ge 0$ for respective constraints. We update them iteratively using constraint violations (dual ascent/subgradient), and at convergence, their values signify the constraints which are active and how strongly they influence the optimal semantic selection.

Having constructed the Lagrangian, the next step is to characterize how the
primal variables and the Lagrange multipliers interact. Specifically, for any
fixed set of multipliers, the Lagrangian evaluates the primal objective together with linear penalties for constraint violations. To understand how these penalties shape the feasible region and influence the optimal semantic
selection, we introduce the \emph{dual function}, obtained by minimizing the
Lagrangian over the primal variables. The dual function captures the lower bound on the primal objective attainable for a given choice of multipliers, and forms the basis for the dual maximization problem \cite{boyd2004convex}.

\subsubsection*{Dual Function} 
The dual function is defined as
\begin{equation}
\label{Dual F}
g(\lambda, {\mu}, {\nu}, {\eta}, {\theta}) = \min_{0 \le s_i \le b_i} \mathcal{L}({s}, \lambda, {\mu}, {\nu}, {\eta}, {\theta})  
\end{equation}
which provides the lower bound of the primal objective for fixed dual variables. Thus, $g$ measures how well the multipliers penalize constraint 
violations.

The dual function defined in~\eqref{Dual F} is concave in its arguments, regardless of whether the primal problem is convex\cite{boyd2004convex}. This property enables the use of efficient dual ascent or subgradient methods to iteratively update the Lagrange multipliers associated with multi-constraints. As a result, the proposed framework can handle complex coupled constraints while maintaining computational tractability and convergence guarantees at the dual level.

\subsubsection*{Dual Problem} 
It is known that, the dual function computes the lower bound on the primal objective for lagrange multipliers. Because the dual function is a lower bound, the dual problem seeks the \emph{best} such bound by maximizing $g$. The dual problem is then given as
\begin{equation}
\label{Dual P}
\max_{\substack{\lambda, {\mu}, {\nu}, \\ {\eta}, {\theta} \ge 0}} g(\lambda,{\mu}, {\nu}, {\eta}, {\theta})
\end{equation}
which seeks the tightest (largest) lower bound of the primal objective. The inner minimization over the primal optimization variable ${s}$ and the outer maximization over the multipliers thus naturally arise from lagrangian duality theory. Thus, even though the primal problem is a minimization, the corresponding Lagrangian dual problem is a maximization problem. The algorithm \ref{Algo 3} outlines the Lagrangian dual ascent approach for the relaxed semantic selection problem.

\subsubsection*{Weak and Strong Duality}
For the relaxed semantic selection problem expressed in~\eqref{Cont. Rel}, the primal objective is convex in $s$ and all constraints are affine after relaxation. Therefore, the problem constitutes a convex program. From the standard convex duality theory \cite{boyd2004convex}, weak duality always holds, which implies that the dual function provides a lower bound on the optimal primal objective value. Furthermore, if Slater’s condition \cite{boyd2004convex}, is satisfied (i.e., there exists a strictly feasible $s$ satisfying all inequality constraints with strict inequality), then strong duality holds and the duality gap is zero. In this case, the optimal
primal value equals the optimal dual value, and the primal–dual optimal pair satisfies the KKT conditions.

Strong duality ensures that solving the dual maximization problem is equivalent to solving the relaxed primal semantic selection problem. Consequently, the optimal Lagrange multipliers quantify the sensitivity of the optimal objective to bandwidth, semantic consistency,
information redundancy, cGAN reconstruction, and SNR constraints. This provides both computational tractability and interpretability. Specifically, the multipliers indicate which physical or semantic constraints are active and how strongly they influence optimal semantic resource allocation. We leverage duality to interpret semantic–physical trade-offs and to design tractable algorithms.
\vspace{-4mm}
\subsection{KKT Conditions}
The KKT conditions provide the fundamental optimality characterization for the relaxed semantic selection problem shown in \eqref{Convexity Verification}. For convex problems satisfying the Slater’s condition, these conditions are necessary and
sufficient for optimality. It characterizes stationary points where primal feasibility, dual feasibility, and complementary slackness hold simultaneously. Verifying or enforcing the KKT conditions helps to ensure that the obtained solution is optimal or close to optimal.

Let 
\begin{equation}
r = \Phi(X) - \sum_{i=1}^n s_i u_i,
\label{residual}
\end{equation}
which denotes the semantic representation residual. Differentiating the Lagrangian with respect to each selection variable $s_i$, gives the stationarity condition as shown below
\begin{equation}
\label{eq:kkt_stationarity}
\begin{aligned}
\nabla_{s_i}\mathcal{L} =\;&
-2\lambda_1\, u_i^\top r
+ \lambda_2 c_i
+ \lambda
+ \sum_{j\neq i} \nu_{ij} \\
&\quad
+ \eta_i\!\left(
\mathbb{E}_{z}\!\left[\|x_i - D(G(z|x_i))\|_2^2\right] - \epsilon
\right) \\
&\quad
+ \theta_i\left(\gamma_{\min} - \mathrm{SNR}(x_i)\right)
= 0.
\end{aligned}
\end{equation}
The remaining KKT conditions consist of: (i) Primal feasibility: all relaxed constraints (bandwidth, semantic consistency, information redundancy, cGAN reconstruction, and SNR reliability) must be satisfied by $s$. (ii) Dual feasibility: multipliers $\lambda,\, \mu_i,\, \nu_{ij},\, \eta_i,\, \theta_i \ge 0$. (iii) Complementary slackness: each multiplier must vanish unless its associated constraint is active. Together, these conditions characterize any optimal solution of the relaxed problem and guide the design of iterative primal–dual update methods in the semantic transmission framework. Algorithm \ref{Algo 4} presents the primal–dual optimization framework based on KKT conditions.

\begin{algorithm}
\caption{Primal--dual optimization with KKT conditions}
\label{Algo 4}
\begin{algorithmic}[1]
\REQUIRE Embeddings $U$, target vector $\Phi$, costs $c_i$, bandwidth budget $k_{\max}$, SNR threshold $\gamma_{\min}$, tolerances $\varepsilon_{\mathrm{KKT}}, \varepsilon_{\mathrm{cons}}$, step size $\eta>0$
\STATE Initialize $s^{(0)}$ feasible; multipliers $(\lambda,\mu_i,\nu_{ij},\eta_i,\theta_i)\gets 0$
\FOR{$t=0,1,\dots,T_{\max}$}
  \STATE Compute residual $r^{(t)} = \Phi - U s^{(t)}$
  \STATE \textbf{Primal update:}
  \[
  s^{(t+1)} = \arg\min_{0\le s_i\le b_i,\ \sum s_i\le k_{\max}}
  \mathcal{L}(s,\lambda,\mu,\nu,\eta,\theta)
  \]
  using a projected gradient or QP solver.
  \STATE Update residual $r^{(t+1)} = \Phi - U s^{(t+1)}$
  \STATE Compute constraint violations:
  \[
  \begin{aligned}
  &C_1 = \sum_i s_i - k_{\max}\\
  &C_{4,i} = KL_i - \delta\\
  &C_{5,ij} = s_i + s_j - (1+\tau'_{ij})\\
  &C_{3,i} = (E_i - \epsilon)s_i\\
  &C_{6,i} = s_i(\gamma_{\min} - \mathrm{SNR}(x_i))
  \end{aligned}
  \]
  \STATE \textbf{Stationarity check:} for each $i$, evaluate
  \[
  \begin{aligned}
  \nabla_{s_i}\mathcal{L}
  = &-2\lambda_1 u_i^\top r^{(t+1)}
  + \lambda_2 c_i
  + \lambda
  + \sum_{j\neq i} \nu_{ij} \\
  &\quad
  + \eta_i (E_i - \epsilon)
  + \theta_i (\gamma_{\min} - \mathrm{SNR}(x_i))
  \end{aligned}
  \]
  \STATE \textbf{Dual ascent:}
  \[
  \begin{aligned}
  &\lambda \gets [\lambda + \eta\, C_1]_+\\
  &\mu_i \gets [\mu_i + \eta\, C_{4,i}]_+\\
  &\nu_{ij} \gets [\nu_{ij} + \eta\, C_{5,ij}]_+\\
  &\eta_i \gets [\eta_i + \eta\, C_{3,i}]_+\\
  &\theta_i \gets [\theta_i + \eta\, C_{6,i}]_+
  \end{aligned}
  \]
  \IF{$\max_i |\nabla_{s_i}\mathcal{L}| \le \varepsilon_{\mathrm{KKT}}$ \textbf{and} all constraint violations, $C\le \varepsilon_{\mathrm{cons}}$}
      \STATE \textbf{stop}
  \ENDIF
\ENDFOR
\STATE \textbf{Rounding:} convert relaxed $s$ to binary by thresholding or top-$k$, , then repair to satisfy bandwidth budget and diversity.
\STATE \textbf{Return} feasible solution $S=\{i : s_i=1\}$
\end{algorithmic}
\end{algorithm}

Overall, the methodological steps detailed in this section convert the discrete combinatorial optimization problem into a continuous tractable framework, leverage the submodular structure for efficient approximation and efficient algorithms, incorporate complex constraints via Lagrangian duality, and establish rigorous optimality conditions. Together, they provide a theoretical and practical framework for obtaining better quality, provably near optimal solutions to the SemCom optimization problem.
\vspace{-4mm}
\section{Simulation Results and Analysis}
This section presents a detailed analysis of the performance evaluation for the proposed SA--cGAN model. The text dataset used in our work consisted of sentences ranging from 5 to 30 words, with a vocabulary size of 6,267. The performance of the proposed framework is evaluated across multiple channel conditions and compared with Huffman+Turbo \cite{Qin2022Semantic}, joint source--channel coding (JSCC) \cite{8461983}, and a fixed-length (5-bit) + Reed Solomon (RS), Huffman+RS, 5-bit + Turbo \cite{Xie2021} coding baselines. Fig.~\ref{BLEU} illustrates the BLEU performance of different communication schemes as a function of SNR. At low SNR regimes (0--5 dB), the proposed SA--cGAN achieves BLEU scores in the range of \(\approx 0.01\text{--}0.03\), which are comparable to Huffman + Turbo (\(\approx 0.02\text{--}0.03\)), 5 Bit + RS (\(\approx 0.02\text{--}0.04\)), and 5 Bit + Turbo (\(\approx 0.01\text{--}0.03\)), but significantly lower than JSCC, which maintains a relatively stable BLEU score of \(\approx 0.42-0.46\) in this region. At moderate SNR values (7.5--10 dB), SA--cGAN shows a gradual increase in BLEU from \(\approx0.04\) to \(\approx0.1\), closely following the performance of traditional source--channel coding schemes, while JSCC remains nearly constant at \(\approx 0.46\text{--}0.47\). As the SNR exceeds 12.5 dB, the semantic advantage of SA--cGAN becomes more pronounced. At 15 dB, SA--cGAN achieves a BLEU score of \(\approx 0.43\), outperforming Huffman + Turbo (\(\approx 0.26\)), 5 Bit + Turbo (\(\approx 0.16\)), 5 Bit + RS (\(\approx 0.31\)) and Huffman + RS (\(\approx 0.36\)). In the high SNR regime (17.5--20 dB), SA--cGAN shows a significant improvement, reaching BLEU scores of \(\approx 0.61\) at 17.5 dB and \(\approx 0.72\) at 20 dB, thus surpassing Huffman + Turbo (\(\approx 0.43\text{--}0.53\)), 5 Bit + RS (\(\approx 0.46\text{--}0.55\)), Huffman + RS (\(\approx 0.49\text{--}0.60\)), and JSCC (\(\approx 0.46\text{--}0.47\)). These results highlight that BLEU, which measures structural n-gram similarity, is not the primary strength of SA--cGAN, as the system prioritizes semantic meaning preservation rather than exact syntactic reconstruction.

\begin{figure}
	\begin{center}
		\includegraphics[width=1\linewidth]{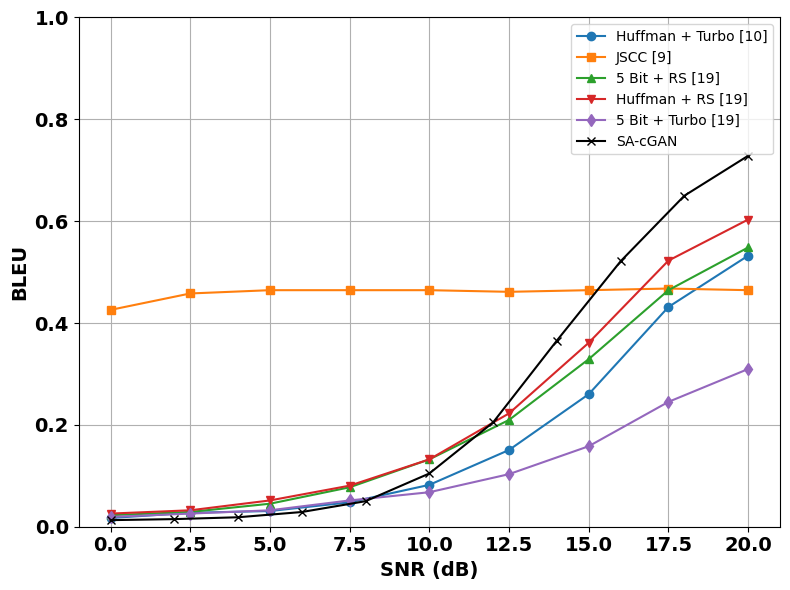}
        \vspace{-6mm}
		\caption[System model]{BLEU score comparison of the proposed SA-cGAN model with existing approaches}
        \vspace{-6mm}
		\label{BLEU}
	\end{center}
\end{figure}

Fig.~\ref{SEM ACC/COM} shows that both semantic accuracy (token-level precision) and semantic completeness (token-level recall) follow a smooth and monotonic trend with increasing SNR. At 0~dB, the system achieves a semantic accuracy of \(\approx 0.19\) and a semantic completeness of \(\approx 0.18\), indicating that even under severe noise conditions, the core semantic elements remain partially recoverable. As SNR increases, both semantic accuracy and semantic completeness improve monotonically, reaching high values at larger SNR levels, which confirms improved semantic reconstruction under reliable channel conditions. The tight coupling of the two curves confirms that the proposed SA--cGAN model effectively balances token-level precision (semantic accuracy) and token-level recall (semantic completeness). This behavior reflects the effectiveness of the self-attention-based semantic extractor in prioritizing semantically important tokens and the role of KG-aligned representations in preserving contextual relationships across varying channel conditions and cGAN reconstruction constraint. Notably, the absence of large gaps between semantic accuracy and semantic completeness curve indicates minimal semantic hallucination and limited loss of essential information, even when aggressive semantic compression is applied. Compared to BLEU, these metrics demonstrate that SA--cGAN preserves meaning more reliably than surface linguistic structure, aligning precisely with the goals of SemCom.
\begin{figure}
	\begin{center}
		\includegraphics[width=1\linewidth]{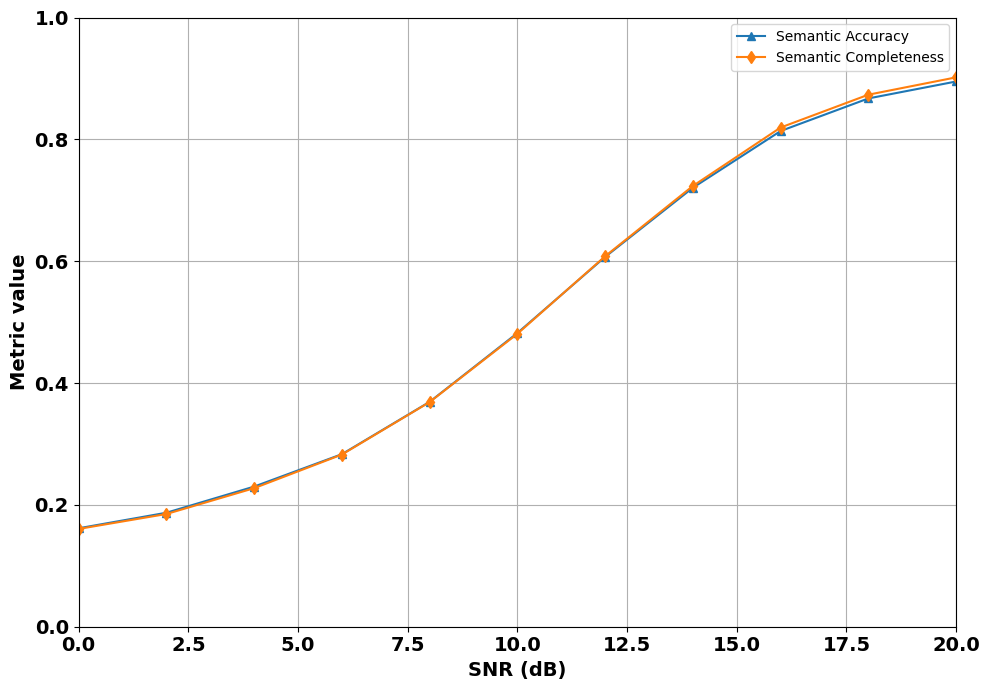}
        \vspace{-6mm}
		\caption[System model]{Semantic accuracy and semantic completeness versus SNR of the proposed SA–cGAN model}
        \vspace{-8mm}
		\label{SEM ACC/COM}
	\end{center}
\end{figure}
Fig.~\ref{Sentence Similarity} illustrates the semantic similarity performance as a function of SNR using contextual BERT embeddings. Unlike BLEU, the semantic similarity metric captures the meaning, synonyms, and contextual semantic alignment. The results show that semantic similarity increases with SNR in all methods, reflecting improved semantic reconstruction under better channel conditions. In the low SNR regime, SA--cGAN significantly outperforms traditional coding schemes, which exhibit near-complete semantic degradation, while JSCC achieves moderate performance. As SNR increases, SA--cGAN continues to improve steadily due to its adaptive semantic representation mechanism, eventually surpassing baseline methods and achieving near-perfect semantic preservation at high SNR. Unlike conventional approaches, which focus on symbol recovery, the proposed method enables robust and adaptive meaning preservation, demonstrating its suitability for SemCom in 6G systems. Although JSCC achieves higher semantic similarity in certain SNR regimes, it operates by transmitting dense representations without explicit semantic selection. In contrast, the proposed SA–cGAN framework performs adaptive semantic compression by selecting and transmitting only the most informative semantic components based on information representability and semantic reliability. This enables significant reductions in bandwidth usage while preserving core semantic content. Furthermore, SA–cGAN provides an interpretable and controllable mechanism to balance semantic fidelity and transmission cost, which is not achievable in conventional JSCC approaches. Therefore, the value of the proposed method lies not only in the performance of semantic similarity but also in its ability to achieve efficient, adaptive, and meaning-aware communication under resource constraints.

\begin{figure}
	\begin{center}
		\includegraphics[width=1\linewidth]{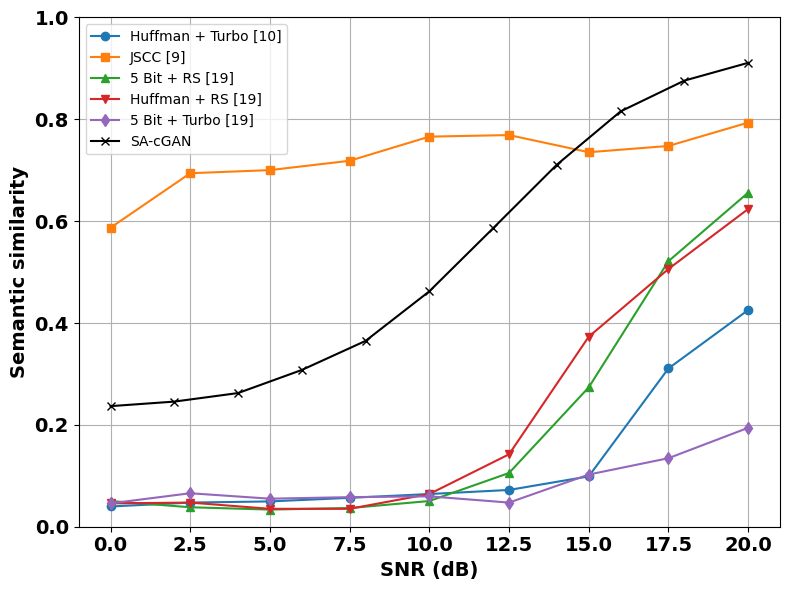}
        \vspace{-6mm}
		\caption[System model]{Semantic similarity comparison of the proposed SA-cGAN model with existing approaches}
        \vspace{-4mm}
		\label{Sentence Similarity}
	\end{center}
\end{figure}
\begin{figure}
    \centering
    \includegraphics[height=12cm,width=8.5cm]{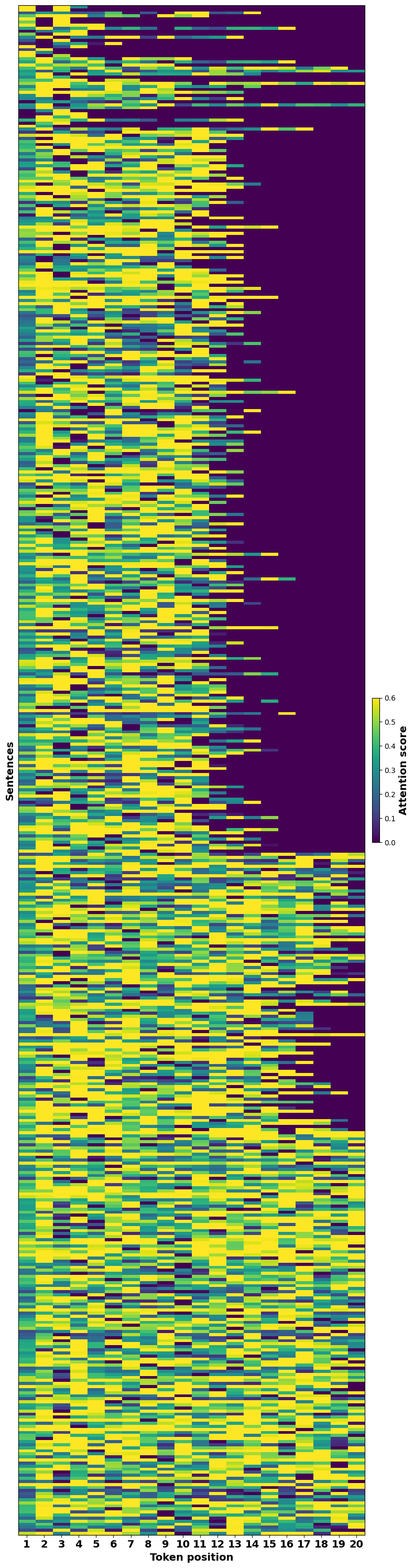}
    \vspace{-4mm}
    \caption{Semantic token importance heatmap for SA‑cGAN}
    \vspace{-6mm}
    \label{Heatmap}
\end{figure}

On the other hand, Fig.~\ref{Heatmap} illustrates the self‑attention behavior of the proposed SA‑cGAN model. The heatmap aggregates token wise importance scores for the 500 sentences (rows) and 20 token positions (columns), where bright yellow cells correspond to high attention scores, and dark cells correspond to tokens that are either padded or deemed semantically unimportant. Across the corpus, the model consistently assigns higher attention to a sparse subset of tokens in each sentence, while suppressing the padded tail region and many low information tokens, confirming that the learned selector performs effective semantic compression rather than uniformly weighting all positions. Another important observation is that high attention regions are spread over different positions and sentences instead of being confined to the first few tokens, which indicates that the self‑attention and sentence‑wise normalization steps have removed the positional bias and drive the selection based on contextual importance. This behavior matches the design objective of SA‑cGAN, the encoder automatically concentrates transmission resources on content words, entities, and key relations that are critical for reconstructing the meaning at the receiver, while allowing redundant or less informative words to be down weighted, when bandwidth is limited. Fig.~\ref{Token selection} shows the distribution of semantic compression ratios achieved by the SA-cGAN model across the evaluated corpus, where the ratio denotes the fraction of selected tokens relative to the total number of tokens in each sentence. The histogram shows that a large portion of sentences exhibit high selection ratios, with \(\approx\) 400--450 sentences retaining more than 80\% of tokens (selection ratio $> 0.8$). This behavior corresponds to a conservative semantic adaptation regime, where distortion criticality is high and therefore require preserving most semantic elements to maintain semantic fidelity. A substantial number of sentences \(\approx\) 300--350 samples fall within the moderate compression ratio range of 0.6--0.8 (i.e., 20\%--40\% compression), indicating balanced semantic adaptation. In this regime, the model successfully removes redundant or low importance tokens while retaining core semantics, achieving meaningful bandwidth savings without compromising semantic integrity. Around \(\approx\) 150--300 sentences exhibit compression ratios in the 0.4--0.6 (i.e., 40\%--60\% compression) range, reflecting more aggressive pruning enabled by higher information representability and lower distortion criticality. Only a small fraction of sentences, fewer than \(\approx\) 60--70 samples experience highly aggressive compression with selection ratios below 0.4 (i.e., $>$ 60\% compression), demonstrating that the proposed model rarely discards semantics excessively. This distribution confirms the model’s capability for adaptive bandwidth allocation, longer sentences containing redundant tokens with high information representability are compressed more aggressively, achieving reductions of up to 40--60\%, whereas shorter sentences with high distortion criticality preserve nearly all tokens to ensure semantic fidelity. In general, the histogram validates the proposed framework’s ability to dynamically trade-off bandwidth efficiency and semantic fidelity based on information representability and distortion criticality.
\begin{figure}
	\begin{center}
		\includegraphics[width=1\linewidth]{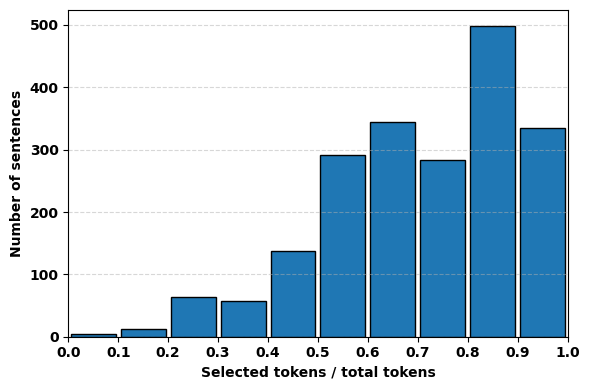}
        \vspace{-6mm}
		\caption[System model]{SA-cGAN semantic compression ratio distribution}
        \vspace{-6mm}
		\label{Token selection}
	\end{center}
\end{figure}

Furthermore, Fig.~\ref{Performance components} presents an SNR dependent analysis of the four key components of the SA-cGAN SemCom pipeline. At low SNR (0~dB), semantic selection exhibits the highest effectiveness ($\approx$ 1.0), highlighting its critical role in prioritizing essential tokens when channel reliability is poor. The KG effectiveness also peaks in this regime, indicating that the knowledge graph is particularly valuable when received semantics are noisy or incomplete, as it compensates for semantic degradation through relational inference. As SNR improves from 6~dB to 12~dB, channel capacity and reconstruction quality also improve, increasing from roughly 0.2 to 0.5, while the effectiveness of semantic selection decreases (from $\approx$ 0.9 to 0.6), and the effectiveness of KG gradually declines. This trend reflects a reduced need for aggressive semantic pruning and repair under cleaner channel conditions. At high SNR (20~dB), the channel capacity and reconstruction quality dominate ($\approx$ 1.0), while the contributions of semantic selection and KG decrease to very low values. This behavior confirms that the relative importance of the knowledge graph decreases as the reliability of the channel improves, demonstrating the adaptive resource allocation of the system in different SNR regimes.
\begin{figure}
	\begin{center}
		\includegraphics[width=1\linewidth]{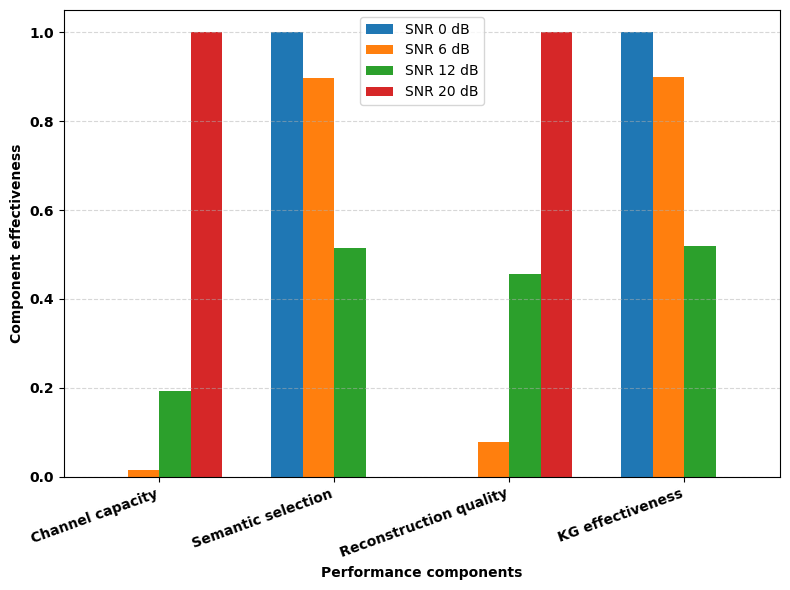}
        \vspace{-6mm}
		\caption[System model]{SNR dependent SA-cGAN component effectiveness}
        \vspace{-6mm}
		\label{Performance components}
	\end{center}
\end{figure}
Finally, Fig.~\ref{Loss} demonstrates the convergence behavior of the loss function for the proposed SA-cGAN model. The loss decreases rapidly in the initial iterations, indicating efficient learning of coarse semantic representations, followed by a gradual stabilization phase reflecting fine-grained refinement. The smooth and stable convergence without divergence demonstrates the effectiveness of the training process. This behavior confirms the robustness of the proposed model for reliable SemCom under practical channel conditions.
\begin{figure}
	\begin{center}
		\includegraphics[width=1\linewidth]{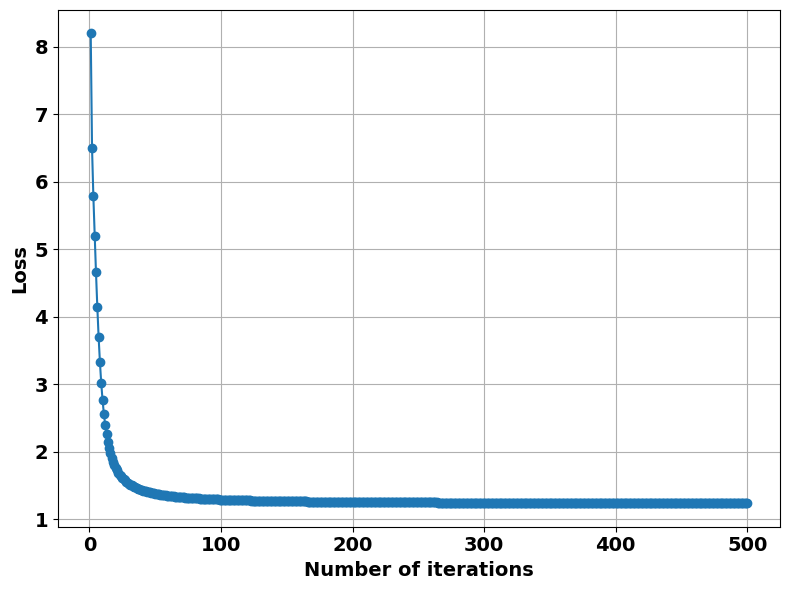}
        \vspace{-6mm}
		\caption[System model]{Training loss behavior of the proposed SA–cGAN model over 500 iterations}
        \vspace{-6mm}
		\label{Loss}
	\end{center}
\end{figure}

Overall, these comprehensive results establish that while SA--cGAN may not always maximize syntactic BLEU metrics, particularly at low SNR, it significantly outperforms existing schemes in semantic oriented metrics, including semantic similarity, semantic accuracy, and semantic completeness. This aligns with the intended design philosophy of SemCom, where preserving meaning under resource and channel constraints is more critical than reconstructing exact symbol sequences. The observed improvements demonstrate the effectiveness of the proposed self-attention extractor, KG-driven semantic representation, and cGAN-powered adaptive transmission strategy in achieving robust, resource-efficient SemCom over noisy channels.
\vspace{-4mm}
\section{Conclusion}
In this work, we proposed an SA–cGAN framework for SemCom that explicitly balances the trade-off between distortion criticality and information representability under varying channel conditions. Unlike conventional symbol-centric systems, the proposed approach adaptively selects and transmits semantic representations according to their contextual importance and vulnerability to channel distortion, enabling efficient utilization of bandwidth. The results show that although SA–cGAN achieves a lower syntactic fidelity (BLEU) performance in low-SNR regimes, it significantly outperforms both traditional and JSCC-based schemes at high SNR. The proposed model consistently achieves high semantic similarity across varying channel conditions, along with progressive improvements in semantic accuracy and semantic completeness metric. These results confirm that SA–cGAN effectively preserves semantic content under severe noise by prioritizing distortion critical semantics in conservative regimes and exploiting higher information representability in favorable channel conditions. Furthermore, the semantic compression analysis reveals adaptive bandwidth allocation behavior, where sentences with high distortion criticality retain most tokens, while highly representable sentences undergo aggressive semantic compression without loss of meaning. The stable convergence of the training loss validates the robustness of the adversarial learning process and the effectiveness of the integrated self-attention and knowledge-graph guidance. Overall, SA-cGAN integrated KG assisted semantic extraction and SNR aware adaptive transmission, enables robust and resource efficient communication that aligns with 6G goals, demonstrating clear advantages over existing symbol-centric schemes and highlighting the transformative potential of semantic aware transmission strategies. Future work will focus on extending the proposed SA–cGAN framework to multi-user and multi-modal SemCom scenarios, including joint processing of textual, visual, and sensor data. Moreover, investigating online adaptive semantic encoding and reinforcement learning based resource allocation under dynamically varying channel conditions represents a promising direction for further enhancing system robustness and efficiency.
\vspace{-4mm}

\end{document}